\documentclass[lettersize, journal]{IEEEtran}
\usepackage{setspace}
\usepackage{cite}
\usepackage{pgfplots}

\pgfplotsset{compat=1.10}

\usepgfplotslibrary{fillbetween}

\usepackage{bm}
\usepackage{lipsum}
\usepackage{makeidx}
\usepackage{enumerate}
\usepackage{color}
\usepackage{cite}
\usepackage{amsmath,amsthm}   
\usepackage{amssymb}
\usepackage{nomencl}
\usepackage{multirow}
\usepackage{graphicx}
\usepackage{epstopdf}
\usepackage{threeparttable}
\usepackage{multicol}
\DeclareGraphicsExtensions{.pdf,.jpeg,.png,.jpg,.emf,.eps}
\usepackage{calc}
\usepackage{here}
\usepackage{url}

\usepackage{upref}
\usepackage{comment}
\usepackage{times}
\usepackage{dsfont}
\usepackage{epic,eepic}
\usepackage{rawfonts}
\usepackage[T1]{fontenc}
\usepackage{latexsym}
\usepackage{amsfonts}

\usepackage{capitalgreekitalic}

\usepackage[font=small,labelfont=bf]{caption}
\usepackage[font=small,labelfont=bf]{subcaption}

\usepackage{xcolor}
\usepackage{tikz,pgfplots}
\usetikzlibrary{calc}
\usetikzlibrary{intersections}
\usetikzlibrary{arrows,shapes}
\usetikzlibrary{positioning}

\newtheorem{corollary}{Corollary}

\newtheorem{lemma}{Lemma}

\theoremstyle{definition}

\allowdisplaybreaks

\begin{document}
\title{Rate Splitting Multiple Access for RIS-aided URLLC MIMO Broadcast Channels
}
\author{Mohammad Soleymani, \emph{Member, IEEE},  
Ignacio Santamaria, \emph{Senior Member, IEEE}, 
Eduard Jorswieck, \emph{Fellow, IEEE}, Marco Di Renzo, \emph{Fellow, IEEE},
Robert Schober, \emph{Fellow, IEEE}, and
Lajos Hanzo, \emph{Life Fellow, IEEE}
 \\ \thanks{ 
Mohammad Soleymani is with the Signal and System Theory Group, University of Paderborn, 33098 Paderborn, Germany (email: \protect\url{mohammad.soleymani@uni-paderborn.de}).  

Ignacio Santamaria is with the Department of Communications Engineering, Universidad de Cantabria, 39005 Santander, Spain (email: \protect\url{i.santamaria@unican.es}).

Eduard Jorswieck is with the Institute for Communications Technology, Technische Universit\"at Braunschweig, 38106 Braunschweig, Germany
(email: \protect\url{jorswieck@ifn.ing.tu-bs.de}).

Marco Di Renzo is with Universit\'e Paris-Saclay, CNRS, CentraleSup\'elec, Laboratoire des Signaux et Syst\`emes, 3 Rue Joliot-Curie, 91192 Gif-sur-Yvette, France. (email: \protect\url{marco.di-renzo@universite-paris-saclay.fr}), and with King's College London, Centre for Telecommunications Research -- Department of Engineering, WC2R 2LS London, United Kingdom (email: \protect\url{marco.di_renzo@kcl.ac.uk}).


Robert Schober is with the Institute for Digital Communications, Friedrich Alexander University of Erlangen-Nuremberg, Erlangen 91058, Germany (email: \protect\url{robert.schober@fau.de}).

Lajos Hanzo is with the Department of Electronics and Computer Science, University of Southampton, SO17 1BJ Southampton, United Kingdom (email:
\protect\url{lh@ecs.soton.ac.uk}).

The work of I. Santamaria was funded by MICIU/AEI /10.13039/501100011033, under Grant PID2022-137099NB-C43 (MADDIE and FEDER, UE), and by European Commission’s Horizon Europe, Smart Networks and Services Joint Undertaking, research and innovation program under grant agreement 101139282, 6G-SENSES project. The work of E. Jorswieck was supported by the Federal Ministry of Education and Research (BMBF, Germany) through the Program of Souver\"an. Digital. Vernetzt. joint Project 6G-RIC, under Grant 16KISK031, and by European Union's (EU's) Horizon Europe project 6G-SENSES under Grant 101139282. The work of M. Di Renzo was supported in part by the European Union through the Horizon Europe project COVER under grant agreement number 101086228, the Horizon Europe project UNITE under grant agreement number 101129618, the Horizon Europe project INSTINCT under grant agreement number 101139161, and the Horizon Europe project TWIN6G under grant agreement number 101182794, as well as by the Agence Nationale de la Recherche (ANR) through the France 2030 project ANR-PEPR Networks of the Future under grant agreement NF-PERSEUS 22-PEFT-004, and by the CHIST-ERA project PASSIONATE under grant agreements CHIST-ERA-22-WAI-04 and ANR-23-CHR4-0003-01. Robert Schober’s work was supported by the Federal Ministry for Research, Technology and Space (BMFTR) in Germany in the program of ``Souver\"an. Digital. Vernetzt.'' joint project 6G-RIC (Project-ID 16KISK023) and the Deutsche Forschungsgemeinschaft (DFG, German Research Foundation) under projects SFB 1483 (Project-ID 442419336, EmpkinS). Lajos Hanzo would like to acknowledge the financial support of the Engineering and Physical Sciences Research Council (EPSRC) projects under grant EP/X01228X/1, EP/Y026721/1, EP/W032635/1 and EP/X04047X/1.
}}
\maketitle
\begin{abstract}
The performance of modern wireless communication systems is typically limited by interference. The impact of interference can be even more severe in ultra-reliable and low-latency communication (URLLC) use cases. A powerful tool for managing interference is rate splitting multiple access (RSMA), which encompasses many multiple-access technologies like non-orthogonal multiple access (NOMA), spatial division multiple access (SDMA), and broadcasting. Another effective technology to enhance the performance of URLLC systems and mitigate interference is constituted by reconfigurable intelligent surfaces (RISs). This paper develops RSMA schemes for multi-user multiple-input multiple-output (MIMO) RIS-aided broadcast channels (BCs) based on finite block length (FBL) coding. We show that RSMA and RISs can substantially improve the spectral efficiency (SE) and energy efficiency (EE) of MIMO RIS-aided URLLC systems. Additionally, the gain of employing RSMA and RISs noticeably increases when the reliability and latency constraints are more stringent. Furthermore, RISs impact RSMA differently, depending on the user load. If the system is underloaded, RISs are able to manage the interference sufficiently well, making the gains of RSMA small. However, when the user load is high, RISs and RSMA become synergetic.
\end{abstract}
\section{Introduction}\label{1}
The sixth generation of wireless communication systems (6G) has to substantially improve the spectral efficiency (SE), energy efficiency (EE), peak and average data rate, latency, and reliability \cite{wang2023road, gong2022holographic}. To reach these ambitious goals, 6G has to overcome numerous challenges. In particular, due to a shortage in the available spectrum, we may have to operate in regimes where the number of users is significantly higher than the number of available resource blocks, which leads to strong interference. The impact of interference can be even more severe in ultra-reliable low-latency communication (URLLC) systems, which are needed to guarantee the extra requirements of 6G applications such as cloud gaming or robotic aided surgery. A powerful technique to handle interference is rate splitting multiple-access (RSMA), which encompasses non-orthogonal multiple access (NOMA), spatial division multiple access (SDMA), broadcasting, and multi-casting \cite{mao2022rate}. Another promising technology conceived for managing interference is represented by reconfigurable intelligent surfaces (RISs) \cite{wu2021intelligent, di2020smart}. In this paper, we propose RSMA schemes for enhancing the SE and EE of RIS-aided multi-user (MU) multiple-input multiple-output (MIMO) URLLC systems.

\subsection{Literature Review}\label{sec-i-a}
The optimal treatment of interference critically depends on how strong it is. When the interference is \textit{weak}, treating interference as noise (TIN) is optimal from a generalized degree of freedom (GDoF) or from a sum-rate maximization point of view \cite{geng2015optimality,annapureddy2009gaussian}. By contrast, the optimal approach to deal with \textit{strong} interference is to decode, remodulate and subtract it \cite{sato1981capacity}, which is referred to as successive interference cancellation (SIC). Unfortunately, the optimal decoding strategy is unknown or is very complex for the operational regimes between these two extreme cases, which include many practical scenarios. In these cases, we may often have to utilize interference-management techniques, which are not necessarily optimal. One such effective interference-management technique is RSMA, which can exploit the benefits of both the TIN and SIC schemes \cite{mao2022rate}.  {In the downlink (DL) RSMA\footnote{ {The focus of this work is on the downlink communications. We refer the reader to \cite{mao2022rate} for a distinction between the DL and uplink (UL) RSMA.}}, the total rate of each user can be split into two types of messages, which are referred to as the common and private messages.} The common messages are decoded by a group of (or even all) the users, while the private messages are decoded only by the intended user. Moreover, each user first decodes the corresponding common messages, and then, removes them from the received signal to decode its intended private message. Indeed, to decode the private message, each user applies SIC to the common messages and treats the private messages of the other users as noise. 

 {RSMA is a highly adaptive multiple-access technique that can handle interference  effectively \cite{clerckx2016rate, mishra2022rate, joudeh2016sum, soleymani2022rate}.}
There are various RSMA schemes with different formats for the common messages. The most practical RSMA scheme is the 1-layer rate splitting (RS), which utilizes only a single common message that has to be decoded by all the users \cite{joudeh2016sum, soleymani2022rate, zhou2021rate, mao2021rate, dizdar2021rate, soleymani2023rate, li2020rate, flores2021tomlinson, soleymani2023energy}. 
 {Moreover, the most attractive RSMA scheme is the generalized RS, which features $(K-1)$ layers of message decoding that require a large number of different common messages if the number of users, $K$, is high \cite{mao2018rate}.} Indeed, the complexity of RSMA drastically increases when the number of layers grows, which may affect its practicality. This is even more important for URLLC systems, since decoding multiple layers of common messages may increase the processing latency. Therefore, in this paper, we consider only the 1-layer RS scheme.

Another technology to manage the interference and improve the coverage is RIS, which has been shown to be promising in enhancing the performance of various networks, including the broadcast channel (BC) \cite{huang2019reconfigurable,wu2019intelligent, soleymani2024energy, soleymani2022noma}, $K$-user interference channel (IC) \cite{jiang2021achievable, bafghi2022degrees, xu2023enhanced, liu2024ris, santamaria2023interference, huang2020achievable}, multiple-access channel (MAC) \cite{you2021reconfigurable,fotock2023energy}, cognitive radio systems \cite{zhang2020intelligent}, and URLLC systems \cite{soleymani2023optimization, soleymani2023spectral, soleymani2024optimization}.  {RISs are capable of ameliorating the effective channel, which can help to increase the channel gains for the intended links and/or attenuate the interfering channels. This feature enables RIS to operate as an interference-management technology, especially in ICs \cite{jiang2021achievable, bafghi2022degrees, xu2023enhanced, liu2024ris, santamaria2023interference}. 
For instance, the authors of \cite{jiang2021achievable} proposed interference-neutralizing designs for $K$-user single-input, single-output (SISO) ICs, employing RISs. Moreover, in \cite{santamaria2023interference}, an RIS is used to reduce interference leakage in a $K$-user MIMO IC.}

 {Even though RISs are capable of handling interference in specific scenarios, the performance of systems supporting high user load may be further improved when leveraging RISs along with other interference-management techniques such as improper Gaussian signaling (IGS), NOMA, and RSMA. For instance, the superiority of IGS over proper signaling in RIS-aided systems has been established in \cite{fang2023improper, soleymani2022improper, yu2021maximizing}. Moreover, the authors of \cite{soleymani2022noma, liu2023optimization, li2023ris, soleymani2023noma} showed that NOMA can improve the performance of various types of multiple-antenna RIS-aided systems. Additionally, the
benefits of RSMA in RIS-aided BCs has been investigated in \cite{soleymani2022rate, katwe2022ratetwc, niu2023active, zhang2023joint}. Specifically, the authors of \cite{soleymani2022rate} proposed an optimization framework for 1-layer RS in MIMO RIS-assisted systems and showed that RSMA can enhance the SE and EE of RIS-aided BCs. The authors of \cite{niu2023active} considered the SE and EE tradeoffs in a single-cell multiple-input single-output (MISO) BC with an active RIS and showed that RSMA is capable of enhancing the system performance.}

The studies in \cite{soleymani2022rate, katwe2022ratetwc, niu2023active, zhang2023joint} on RSMA for RIS-aided systems showed that RSMA and RISs can be mutually beneficial. However, these papers considered the asymptotic Shannon rate as the performance metric of interest.  {To realize low latency communications, finite block length (FBL) coding has to be utilized, which makes the asymptotic Shannon rate an inaccurate performance metric.} Specifically, the achievable data rate of FBL coding is lower than the Shannon rate and depends on the channel dispersion, the packet length, and the decoding error probability. Thus, the existing solutions for RSMA in RIS-aided systems based on the Shannon rate cannot be applied to URLLC systems utilizing FBL coding.
 {Unfortunately, there are only a limited number of studies on RSMA in systems employing FBL coding \cite{soleymani2023optimization, wang2023flexible, xu2023max, xu2022rate, singh2023rsma,katwe2022rate, dhok2022rate, pala2023spectral, katwe2024rsma, jorswieck2025urllc}. There are fewer treatises investigating the performance of RSMA in  RIS-aided systems relaying on FBL coding, as shown in Table \ref{table-1}. Furthermore, there is no work studying RSMA in MIMO systems with FBL coding, not even for systems operating without RISs. The authors of \cite{soleymani2023optimization} proposed 1-layer RS schemes for MISO systems assisted by beyond diagonal RISs and showed that RSMA and RISs can substantially enhance the SE and EE of multi-cell MISO BCs. The authors of \cite{katwe2022rate} showed that 1-layer RS is capable of increasing the global EE of a single-cell MISO BC aided by an RIS. In \cite{ singh2023rsma, dhok2022rate}, it was shown that RSMA can improve the performance of different RIS-assisted SISO systems. The papers in \cite{pala2023spectral, katwe2024rsma, jorswieck2025urllc} showed that RSMA can improve the DL SE of a BC with a multiple-antenna BS serving multiple single-antenna users.}

\begin{table}
\centering
\scriptsize
\caption{Comparison of most related works on RSMA in  systems with
FBL coding.}\label{table-1}
\begin{tabular}{|c|c|c|c|c|c|c|c|c|c|c|c|c|c|c|}
	\hline 
 &This paper&\cite{soleymani2023optimization}&\cite{singh2023rsma}&\cite{katwe2022rate}&\!\!\! {\cite{dhok2022rate, pala2023spectral, katwe2024rsma, jorswieck2025urllc}}&\!\!\! {\cite{wang2023flexible, xu2023max, xu2022rate}}
 \\
 \hline
 RIS&$\surd$& $\surd$& $\surd$&$\surd$&$\surd$&-
 \\
 \hline
 Ch. disp. in \cite{scarlett2016dispersion}&$\surd$& $\surd$& -&-&-&-
 \\
 \hline
 SE metrics&$\surd$&$\surd$&$\surd$&-&$\surd$&$\surd$
 \\
 \hline
 EE metrics&$\surd$&$\surd$&-&$\surd$&-&-
 \\
 \hline
 MIMO systems&$\surd$&-&-&-&-&-
 \\
\hline
Multiple streams&$\surd$&-&-&-&-&-
\\
\hline
		\end{tabular}
\normalsize
\end{table}

\subsection{Motivation}
 {MIMO systems play an essential role in modern wireless communication networks. However, \cite{soleymani2024optimization} is the only study on resource allocation for multi-stream  MU-MIMO systems using FBL coding, and there is no literature on RSMA designed for MIMO URLLC systems.} Resource allocations schemes for FBL encoding are much more difficult to handle than the asymptotic Shannon rate and the solutions obtained for the Shannon rate cannot be reused in RSMA URLLC systems. Upon utilizing FBL coding, the achievable data rate depends on the asymptotic Shannon rate, the channel dispersion, the packet length, and the reliability criterion. The channel dispersion is a fractional function of the signal-to-interference-plus-noise ratio (SINR), which makes the optimization of the FBL rate and/or the EE much more challenging than in the asymptotic case. Optimizing the rate for FBL coding is even more complex in MIMO systems, where the optimization variables are matrices.

 {In \cite{soleymani2024optimization}, the authors proposed resource allocation schemes for MU-MIMO RIS-aided URLLC systems when interference is treated as noise, and multiple data streams per user are transmitted. {It was shown that RISs can substantially improve the system performance even with SDMA. However, the benefits of RISs are reduced when the number of users grows. Indeed, the level of interference increases with the number of users, and an RIS alone may not be capable of mitigating the interference by itself.} Thus, in this paper, we aim for integrating RSMA into MIMO RIS-aided URLLC systems.}

 {As shown in Table \ref{table-1}, there is no work on RSMA designed for MIMO systems using FBL coding. Although the authors of \cite{soleymani2023optimization, singh2023rsma,katwe2022rate, dhok2022rate} provided valuable insights into the performance of RSMA in RIS-assisted systems, the algorithms developed in \cite{soleymani2023optimization, singh2023rsma,katwe2022rate, dhok2022rate} cannot be applied to MIMO systems. Indeed, resource allocation for MIMO systems requires more advanced linear algebra tools and is much more complex than that for MISO systems, especially in the FBL case.
Additionally, there are several open research questions, especially for MIMO RIS-aided URLLC systems, that have not been addressed in prior studies and require further investigation. 
In particular, our goal is to answer the following questions:}
\begin{itemize}
    \item   {How does the performance of RSMA and RISs vary with the reliability and latency constraints in MU-MIMO URLLC systems?}

\item   {How do RISs and RSMA impact each other in different operational scenarios?}

\item   {How do the benefits of RISs and RSMA vary for different objective functions, notably SE and EE?}

\item   {Considering RSMA and RIS, which is more beneficial in MU-MIMO URLLC systems when the user load is high?}

\end{itemize}
We show that RSMA and RISs provide higher gains when the reliability and latency constraints are more demanding. Moreover, RSMA further improves the gains provided by RISs when the user load is high. Furthermore, the benefits of RSMA are more significant than those of RISs if the number of users is higher than the number of transmit antennas (TAs) at the base station (BS). 

\subsection{Contributions}
In this study, we develop energy-efficient and spectral-efficient resource allocation schemes for multi-user MIMO URLLC systems by leveraging RISs and RSMA. To the best of our knowledge, this is the first paper to propose RSMA schemes for MU-MIMO systems with FBL coding, where the asymptotic Shannon rate is not a precise metric for characterizing the achievable data rate. We summarize the main contributions made by this paper, as follows: 
\begin{itemize}
    \item  {We show that RSMA substantially increases the SE and EE of MU-MIMO URLLC systems, outperforming the SDMA schemes proposed in \cite{soleymani2024optimization}.  Interestingly, an RIS-aided system employing TIN may perform worse than a system without RISs that uses RSMA. This is particularly the case for systems having high user loads, which are typically interference-limited, and RISs may not be able to fully mitigate the heavy interference.}

    \item  {We show that the performance gains offered by RSMA and RISs increase when the packets and codewords become shorter, and the maximum tolerable bit error rate becomes lower.} Therefore, RSMA can provide higher gains over TIN in URLLC systems, especially when the number of transmit antennas at the BSs is less than  the number of users.

    \item  {We show that RSMA and RISs synergize their benefits in systems with a higher user load. However, when the user load is light, the RISs reduce the benefits of RSMA, since they can effectively mitigate the interference in this case. Hence, there is no need to use sophisticated multiple access schemes in these scenarios.} Additionally, the benefits provided by RISs become more substantial in systems having a high user load, when we utilize RSMA.

    \item  {The gains provided by RSMA become more significant in URLLC RIS-aided systems, when the power budget increases.} Indeed, when the system performance is limited by interference, increasing the transmission power and/or power budget of the transmitters increases the interference level, and thus, it does not noticeably increase the achievable rate and/or EE. However, when we employ RSMA, the interference is appropriately managed by RSMA, which significantly improves the performance at high signal-to-noise ratios (SNR).

    \item  We also develop single-stream RSMA schemes and compare them to the multiple-stream RSMA ones. We show that multiple-stream RSMA significantly outperforms single-stream RSMA. Additionally, switching from MISO to MIMO systems provides significant gains even when there are only a maximum of two data streams per user. Note that even though MISO systems support single-stream data transmissions, we cannot directly use MISO schemes as single-stream MIMO schemes since these solutions are not designed for multiple-antenna receivers. Indeed, the proposed RSMA schemes conceived for MIMO systems substantially outperform an upgraded version of the schemes proposed for MISO systems such as the algorithms in  \cite{soleymani2023optimization}.

\end{itemize}

\subsection{Paper Organization}
The paper is structured as follows. Section \ref{sec=ii} presents the network scenario, the channel models associated with RISs, and the signal model. Moreover, Section \ref{sec=ii} derives the rate and EE expressions, and formulates the optimization problems of interest. Section \ref{sec-iii} develops the algorithms proposed for SE and EE maximization. Section \ref{sec-iv} provides numerical results, while Section \ref{sec-v} summarizes our findings and concludes the paper.

\section{System model}\label{sec=ii}
We consider a multicell MIMO RIS-aided DL BC having $L$ BSs, each equipped with $N_{BS}$ transmit antennas (TAs),  {as shown in Fig. \ref{Fig-sys-model}}.
Each BS serves $K$ users having $N_u$ receive antennas (RAs) each.  {Note that each user is served by only a single BS.} Moreover, $M$ reflective RISs with $N_{RIS}$ elements each aid the BSs. The data packets are assumed to have a finite block length $n$.
 {Moreover, we assume perfect instantaneous global channel state information (CSI), as is common in many studies on RISs \cite{ huang2019reconfigurable, wu2019intelligent, soleymani2022noma, pan2020multicell, soleymani2022improper, soleymani2022rate}. This assumption is also widely used in resource allocation solutions for URLLC systems \cite{nasir2020resource, wang2023flexible,  pala2022joint, ghanem2020resource, soleymani2023spectral}, specifically in systems with high channel coherence times. In such systems, the CSI estimation is easier and more accurate, allowing resource allocation solutions to be reused over multiple time slots with low pilot overhead. Additionally, studying RISs with perfect CSI reveals fundamental system tradeoffs and provides an upper performance bound.}

\begin{figure}[t!]
    \centering
\includegraphics[width=.44\textwidth]{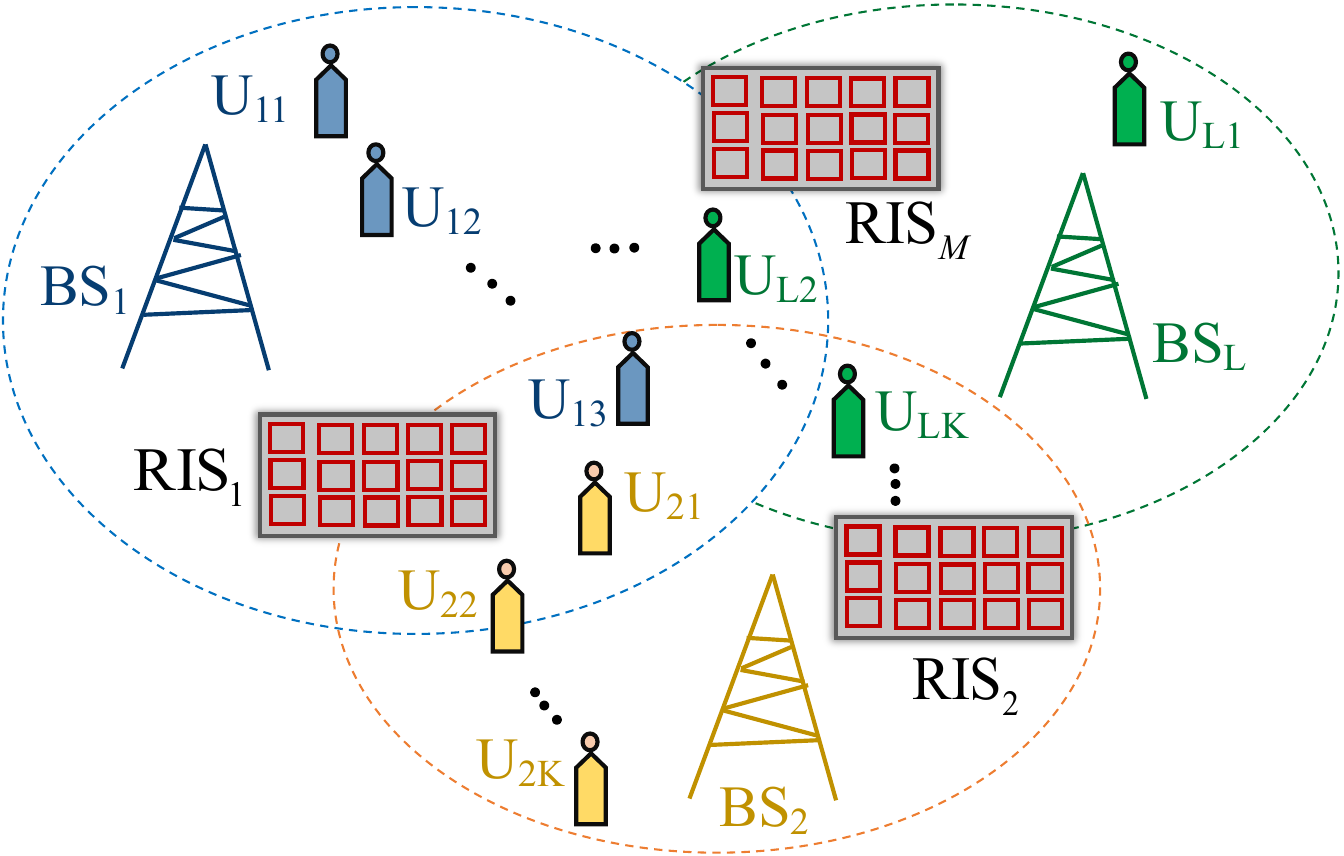}
     \caption{ {A multicell BC assisted by RISs.}}
	\label{Fig-sys-model}
 \end{figure}

\subsection{RIS Model}
We employ the RIS model of \cite{soleymani2022improper,pan2020multicell}, for the MIMO multicell BC considered. Therefore, the channel between BS $i$ and the $k$-th user in cell $l$, represented by U$_{lk}$, is given by
\begin{equation}\label{ch-equ}
\mathbf{H}_{lk,i}\!\left(\{\bf{\Upsilon}\}\right)\!\!=\!\!\! 
\underbrace{\sum_{m=1}^M\!\!\mathbf{G}_{lk,m}{\bf \Upsilon}_m\mathbf{G}_{m,i}}_{\text{Link through RIS}}
+\!\!\!\!
\underbrace{\mathbf{F}_{lk,i}}_{\text{Direct link}}
\!\!\!\!
\in\mathbb{C}^{N_{u}\times N_{BS}}
\!,
\end{equation}
where ${\bf F}_{lk,i}$ is the channel between the $i$-th BS and U$_{lk}$, ${\bf G}_{lk,m}$ is the channel between the $m$-th RIS and U$_{lk}$, ${\bf G}_{m,i}$ is the channel between the $i$-th BS and the $m$-th RIS, $\{{\bf \Upsilon}\} = \{{\bf \Upsilon}_m:\forall m\}$ groups all the RIS coefficients, where ${\bf \Upsilon}_m$ is the scattering matrix for the $m$-th RIS. In this paper, we consider diagonal RISs. Hence, we have
\begin{equation}
{\bf\Upsilon}_m=\text{diag}\left(\upsilon_{m1}, \upsilon_{m2},\cdots,\upsilon_{m{N_{RIS}}}\right),
\end{equation}
where $\upsilon_{mn}$ is the $n$-th reflection coefficient of the $m$-th RIS. The coefficients $\upsilon_{mn}$ for all $m$, $n$ are optimization variables, which control the propagation channels. When an RIS element operates in a nearly passive mode, the power of its output signal is equal to or less than the power of its input signal. Thus, we have $|\upsilon_{mn}|^2 \leq 1$. In this case, the set encompassing all the feasible $\upsilon_{mn}$ is \cite[Eq. (11)]{wu2021intelligent}
\begin{equation}
\mathcal{F}_{U}=\left\{\upsilon_{mn}:|\upsilon_{mn}|^2\leq 1 \,\,\,\forall m,n\right\}.
\end{equation}
which is a convex set. Another common assumption for the RIS elements is the unit modulus constraint $|\upsilon_{mn}| = 1$, which yields the non-convex set \cite[Eq. (47)]{soleymani2022rate}
\begin{equation}
\mathcal{F}_{I}=\left\{\upsilon_{mn}:|\upsilon_{mn}|= 1 \,\,\,\forall m,n\right\}.
\end{equation}
Hereafter, we represent the feasibility set of $\{{\bf \Upsilon}\}$ by $\mathcal{F}$, when we do not refer to a specific set. Furthermore, we drop the dependence of the channels on $\{{\bf \Upsilon}\}$ to simplify the notation.

\subsection{Signal Model}
  {We assume that each BS employs 1-layer RS to serve its associated users, which means that the transmit signal of BS $l$ is given by}
\begin{equation}
{\bf x}_l=
\underbrace{{\bf W}_l{\bf s}_{l,c}}_{\text{Common M.}}+
\underbrace{\sum_k{\bf W}_{lk}{\bf s}_{lk}}_{\text{Private M.}} \in \mathbb{C}^{ N_{BS}\times 1},
\end{equation}
where ${\bf s}_{l,c}\in \mathbb{C}^{ N_{BS}\times 1}$ is the common message of BS $l$, which is decoded by all of its associated users, ${\bf s}_{lk}\in \mathbb{C}^{ N_{BS}\times 1}$ is the private message intended for U$_{lk}$, and the matrices ${\bf W}_l\in \mathbb{C}^{ N_{BS}\times N_{BS}}$ and ${\bf W}_{lk}\in \mathbb{C}^{ N_{BS}\times N_{BS}}$ are, respectively, the beamforming matrices for the common message ${\bf s}_{l,c}$ and the private message ${\bf s}_{lk}$, intended for U$_{lk}$. The signals ${\bf s}_{l,c}$ and ${\bf s}_{lk}$ are assumed to be zero-mean, independent, complex, and proper Gaussian vectors for all $l, k$, with the identity matrix as their covariance matrix. Thus, the covariance matrix of ${\bf x}_l$ is \begin{equation}\label{(6)}
{\bf C}_l\!=\!\mathbb{E}\{{\bf x}_l{\bf x}_l^H\}\!
=\!{\bf W}_l{\bf W}_l^H\!\!+\!\!\sum_k\!{\bf W}_{lk}{\bf W}_{lk}^H\!\in\! \mathbb{C}^{ N_{BS}\times N_{BS}}\!,
\end{equation}
where $\mathbb{E}\{ {\bf X}\}$ denotes the mathematical expectation of ${\bf X}$. Note that the transmit power of BS $l$ is given by $\text{Tr}({\bf C}_l)$, where $\text{Tr}({\bf X})$ denotes the trace of the square matrix ${\bf X}$.  We refer the reader to \cite[Section II.B]{mao2022rate} for more details on the signaling model of the 1-layer RS. We denote the set containing all the precoding matrices as $\{{\bf W}\}=\{{\bf W}_l,{\bf W}_{lk}:\forall l,k\}$.

The received signal at U$_{lk}$ is 
\begin{multline}
{\bf y}_{lk}=\sum_i{\bf H}_{lk,i}{\bf x}_i+{\bf n}_{lk}
=\!\!
\underbrace{
{\bf H}_{lk,l}{\bf W}_{l}{\bf s}_{l,c}
}_{\text{Common M.}}
+
\underbrace{
{\bf H}_{lk,l}{\bf W}_{lk}{\bf s}_{lk}
}_{\text{Private M.}}
\!
\\+
\underbrace{
{\bf H}_{lk,l}\sum_{j\neq k}{\bf W}_{lj}{\bf s}_{lj}
}_{\text{Intracell Interference}}
\!+\!
\underbrace{
\sum_{i\neq l}{\bf H}_{lk,i}{\bf x}_{i}
}_{\text{Intercell Interference}}
\!+\!
\underbrace{{\bf  n}_{lk}
}_{\text{Noise}}\in\mathbb{C}^{N_u\times 1},
\end{multline}
 {where ${\bf n}_{lk}\sim \mathcal{CN}({\bf 0}, \sigma^2 {\bf I} )$ is the zero-mean proper additive complex Gaussian noise at the receiver of U$_{lk}$, $\sigma^2$ is the noise variance at each receive antenna, and ${\bf I}$ is the identity matrix.}

Note that there is both intercell and intracell interference in the considered multi-cell system. The intercell interference is treated as noise, since decoding and canceling the intercell interference may require a large amount of signaling between the BSs \cite{soleymani2022rate}. Moreover, to apply SIC to the intercell interference may be very inefficient, because the BSs should transmit data at a rate that is supportable by outer-cell users. Hence, the data rate becomes limited by the links between the BS and the outercell users, which are typically much weaker than the links between the BS and the inner-cell users. 

\subsection{Channel Dispersion, Rate and EE Expressions}
We assume that the BSs employ FBL coding to support low-latency communication. Thus, the achievable data rate depends on the channel dispersion, the packet length, the reliability constraint, and on the traditional Shannon rate. In this subsection, we first formulate the achievable data rate in multi-user MIMO systems with FBL coding in Lemma \ref{lem-r} and then derive the achievable rate and EE of each user.
\begin{lemma}[\!\cite{polyanskiy2010,soleymani2024optimization}]\label{lem-r}
Upon utilizing the normal approximation (NA), the FBL rate of a user in an interference-limited multi-user MIMO system is
\begin{equation}\label{1-multi}
r=
\underbrace{\log \left|{\bf I} +{\bf D}^{-1}{\bf S} \right|}_{\text{\small Shannon Rate}}
-Q^{-1}(\epsilon)\sqrt{\frac{2\text{\em Tr}({\bf S}({\bf D}+{\bf S})^{-1})}{n}},
\end{equation}
 where $|{\bf X}|$ denotes the determinant of ${\bf X}$, $n$ is the packet length, $\epsilon$  is the maximum tolerable bit error rate for the message, $Q^{-1}$
 is the inverse of the Gaussian $Q$-function, ${\bf S}$ is the covariance  matrix of the desired signal, and ${\bf D}$ is the covariance matrix of the interference plus noise. Moreover, an achievable channel dispersion\footnote{The channel dispersion in Lemma \ref{lem-r} is suboptimal, but it is achievable by Gaussian signaling in multi-user systems in the presence of interference \cite{scarlett2016dispersion}. The optimal channel dispersion is $\text{Tr}({\bf I}- ({\bf I}+ {\bf D}^{-1}{\bf S})^{-2})$ \cite{polyanskiy2010channel}, but it is not achievable by Gaussian signals when there exists interference. Hence, in this paper, we adopt the achievable channel dispersion derived in \cite{scarlett2016dispersion}.} for the considered communication link is $2\text{\em Tr}({\bf S}({\bf D}+{\bf S})^{-1})$.
\end{lemma}
\begin{proof}
    Please refer to \cite[Theorem 68]{polyanskiy2010} and Lemma 1 and
 Lemma 2 in \cite{soleymani2024optimization}.
\end{proof}
In the 1-layer RS scheme, U$_{lk}$ first decodes the common message ${\bf s}_{l,c}$ while treating the other received signals as noise. Hence, the maximum decoding rate of ${\bf s}_{l,c}$ at U$_{lk}$ is
\begin{multline}\label{eq-9}
    r_{c,lk}=
\log \left|{\bf I} +{\bf D}^{-1}_{c,lk}{\bf S}_{c,lk} \right|
\\
-Q^{-1}(\epsilon_c)\sqrt{\frac{2\text{Tr}({\bf S}_{c,lk}({\bf D}_{c,lk}+{\bf S}_{c,lk})^{-1})}{n}},
\end{multline}
 where $n$ is the packet length of the common message, $\epsilon_c$ is  the maximum tolerable bit error rate for the common message,  ${\bf S}_{c,lk}={\bf H}_{lk,l}{\bf W}_l({\bf H}_{lk,l}{\bf W}_l)^H$ is the covariance matrix of the common message at the receiver of U$_{lk}$, and ${\bf D}_{c,lk}$ is the interference-plus-noise covariance matrix  when  decoding the common message ${\bf S}_{c,lk}$, which is given by
\begin{equation}
    {\bf D}_{c,lk}\!=\!\sigma^2{\bf I}\!+\!\!\sum_j\!{\bf H}_{lk,l}{\bf W}_{lj}{\bf W}_{lj}^H{\bf H}_{lk,l}^H\!+\!\!\sum_{i\neq l}\!{\bf H}_{lk,i}{\bf C}_{i}{\bf H}_{lk,i}^H,\!
\end{equation}
 where ${\bf C}_{i}$ is a convex function of ${\bf W}_i$ and ${\bf W}_{ij}$ for all $i,j$, and is given by \eqref{(6)}.

The signal ${\bf s}_{l,c}$ should be transmitted at a rate $r_{cl}$ that can be decoded by all the users associated with BS $l$. Therefore, the maximum transmission rate of ${\bf s}_{l,c}$ is given by
\begin{equation}\label{11}
    r_{c,l}=\min_k(r_{c,lk}).
\end{equation}
The constraint in \eqref{11} is referred to as the decodability constraint for the common message. Note that a user is in outage if it cannot decode the common message. Moreover, the user is unable to decode its own private message without successfully decoding the common message.

 After decoding ${\bf s}_{l,c}$, U$_{lk}$ removes ${\bf s}_{l,c}$ from its received signal and then decodes its own private message. Hence, the maximum achievable rate of ${\bf s}_{p,lk}$ at U$_{lk}$ is given by
\begin{equation}   \label{12}
    r_{p,lk}=
\log \left|{\bf I} +{\bf D}^{-1}_{lk}{\bf S}_{lk} \right|
-Q^{-1}(\epsilon_p)\sqrt{\frac{2\text{Tr}({\bf S}_{lk}{\bf D}_{c,lk}^{-1})}{n}},
\end{equation}
where $n$ is the packet length of the private message, $\epsilon_p$ is  the maximum tolerable bit error rate for the private message,  ${\bf S}_{lk} = {\bf H}_{lk,l}{\bf W}_{lk}({\bf H}_{lk,l}{\bf W}_{lk})^H$ is the covariance matrix of  the private message for U$_{lk}$ at its receiver, and ${\bf D}_{lk}$ is the  covariance matrix of the intracell and intercell interference  plus noise, which is given by
\begin{equation}
    {\bf D}_{lk}\!=\!\sigma^2{\bf I}\!+\!\sum_{j\neq k}{\bf H}_{lk,l}{\bf W}_{lj}{\bf W}_{lj}^H{\bf H}_{lk,l}^H\! +\! \sum_{i\neq l}{\bf H}_{lk,i}{\bf C}_{i}{\bf H}_{lk,i}^H.\!
\end{equation}
Note that ${\bf D}_{c,lk} = {\bf D}_{lk} + {\bf S}_{lk}$. Thus, we can write the achievable channel dispersion of the private message for U$_{lk}$  as $2\text{Tr}\left( {\bf S}_{lk}{\bf D}_{c,lk}^{-1}\right)$ . Additionally, we assume that the common  and private messages have the same finite block length $n$.

 The achievable rate of U$_{lk}$ is equal to the sum of its rate portion for the common message and the rate of its private message as
\begin{equation}
    r_{lk}=t_{lk}+r_{p,lk},
\end{equation}
where $t_{lk}$ is the rate portion of the common message ${\bf s}_{l,c}$ that  is dedicated to U$_{lk}$. Note that the rates ${\bf t} = [t_{lk}:\forall l,k]$ are  indeed optimization variables, and we have $\sum_k t_{lk}\leq r_{c,l}$.  Depending on the choice of ${\bf t}$, 1-layer RS can switch between SDMA, NOMA, and broadcasting when $K = 2$. If  $t_{lk} = 0$ for all $l,k$, the intracell interference is treated as noise, and the 1-layer RS scheme operates as SDMA. Moreover,  when $r_{p,lk} = 0$ for all $l, k$, all the users get their data rate only  from the corresponding common message, which is known as  broadcasting. For further discussions on possible solutions for the common and private rates, please refer to \cite{clerckx2019rate, mao2022rate}.

 Finally, the EE of U$_{lk}$ is \cite{zappone2015energy}
 \begin{equation}\label{(15)}
     e_{lk}=\frac{r_{lk}}{p_c+\eta\text{Tr}\left({\bf W}_{lk}{\bf W}_{lk}^H\right)
     +\frac{\eta}{K} \text{Tr}\left({\bf W}_{l}{\bf W}_{l}^H\right)
     },
 \end{equation}
 where 
$\eta^{-1}$ is the power efficiency of the BSs, and $p_c$ is the constant power consumption of the system, when transmitting data to a user, as given by \cite[Eq. (27)]{soleymani2022improper}.
 
\subsection{Discussions on the Latency and Reliability Constraint }\label{sec-ii-d}
 {In the 1-layer RS scheme, each user should first decode a common message and then apply SIC to the common message to decode its private message. Hence, if the common message is not decoded correctly, then the private message cannot be decoded either. It means that the probability that an error occurs  is
\begin{equation}
    \epsilon=\epsilon_c+(1-\epsilon_c)\epsilon_p \leq \epsilon_c+\epsilon_p.
\end{equation}}

To achieve a latency of less than $\tau$ seconds, the achievable rate of U$_{lk}$ has to be higher than or equal to $r^{th}_{lk}=\frac{2n}{\omega\tau}$ (b/s/Hz), where $\omega$ is the channel bandwidth. Note that we assume that the packet lengths for the common and private messages are equal. Thus, we consider $2n$ as the worst case for the packet length so that $r_{lk}>r^{th}_{lk}$ ensures both the common and private messages are successfully decoded within $\tau$ seconds. Hence, the latency constraint is equivalent to an achievable rate constraint. It is worth emphasizing that this modeling of the reliability and latency constraints aligns with  other studies on URLLC systems such as \cite[Remark 1]{liu2023energy}, \cite[Sec. II.D]{soleymani2023optimization}, \cite{he2021beamforming, soleymani2024optimization}.

\subsection{Problem Statement}
We consider both SE and EE maximization  problems. To enhance the SE, we maximize the minimum weighted rate among all users, which leads to the following optimization problem
\begin{subequations}\label{17}
\begin{align}
\label{17a}
 \underset{r,\{\mathbf{W}\},\{\bf{\Upsilon}\}\in\mathcal{F},{\bf t}
 }{\max}\,\,\,  & 
  r \\
\label{17b}  \text{s.t.}\,\,   \,&  r_{lk}
  \geq \max(\alpha_{lk}r,r_{lk}^{th}),\,\forall l,k,
  \\ 
\label{17c}  & \text{Tr}(\mathbf{C}_l)\leq P_l,\,\, \forall l,
\\
\label{17d}
 &
t_{lk}\geq 0,\,\,\,\,\forall l,k,
\\
\label{17e}
&
\sum_kt_{lk}\leq \min_k(r_{c,lk}),\,\,\forall l,
 \end{align}
\end{subequations}
where $P_l$ is the power budget of BS $l$, the coefficients $\alpha_{lk}$ for all $l, k$ are the weights of the users reflecting their priorities, \eqref{17b} represents the latency constraint as discussed in Section \ref{sec-ii-d}, \eqref{17c} is the power budget constraint, \eqref{17d} is due to the fact that the rates cannot be negative, and \eqref{17e} is the decodability constraint of the common message.  {Additionally, the optimization variables are $r$, $\{\mathbf{W}\}$, $\{\bf{\Upsilon}\}$, and ${\bf t}$, where $r$ is an auxiliary variable, denoting the minimum weighted rate of users.} We, moreover, maximize the minimum weighted EE among users, which leads to
\begin{subequations}\label{18}
\begin{align}
\label{18a}
 \underset{e,\{\mathbf{W}\},\{\bf{\Upsilon}\}\in\mathcal{F},{\bf t}
 }{\max}\,\,\,  & 
  e \\
\label{18b}  \text{s.t.}\,\,   \,&  e_{lk}
  \geq \lambda_{lk}e,\,\forall l,k,
  \\
\label{18c}  &  r_{lk}
  \geq r_{lk}^{th},\,\forall l,k,
  \\ 
&\eqref{17c},\eqref{17d},\eqref{17e},   \end{align}
\end{subequations}
 where the coefficients $\lambda_{lk}$ for all $l, k$ are the weights of the users assigned based on their priorities,  \eqref{18c} is the latency constraint, and $e$ is an auxiliary variable, denoting the minimum weighted EE of users. Note that the algorithms proposed can also be applied to other optimization problems such as maximizing sum rate, global EE and geometric mean of the rates. In this paper, we consider maximization of the minimum weighted rate/EE, since the achievable rate/EE region can be calculated by solving and employing the rate/EE profile technique \cite{zeng2013transmit, zhang2010cooperative, soleymani2019energy, buzzi2016survey, soleymani2020rate}. Additionally, the minimum rate/EE can also be viewed as a metric for the fairness among the users \cite{zhu2023max}.

Hereafter, we refer to the maximum minimum weighted rate (or EE) as the max-min weighted rate (or EE). The optimization problems \eqref{17} and \eqref{18} are complex non-convex problems, since the rates are not jointly concave functions of $\{{\bf W}\}$ and $\{{\bf \Upsilon}\}$. In the next section, we propose efficient algorithms for solving \eqref{17} and \eqref{18}.

\section{Proposed Solutions for Spectral and Energy Efficiency Maximization}\label{sec-iii}
In this section, we derive solutions for \eqref{17} and \eqref{18} by utilizing powerful numerical optimization tools such as alternating optimization (AO), majorization minimization (MM), and the generalized Dinkelbach algorithm (GDA). Our solution is iterative, and starts with a feasible initial point\footnote{ {To find a feasible initial point, we can choose a random point, satisfying the power constraint for $\{{\bf W}\}$ and the unit modulus constraint for $\{{\bf \Upsilon}\}$. If this random point that satisfies \eqref{17c} and belongs to the set $\mathcal{F}_I$ does not satisfy the latency constraint in \eqref{18c}, we can run our algorithm to maximize the minimum rate of users for a few iterations to attain an initial point, satisfying all the constraints in \eqref{17} or \eqref{18}.}} denoted as $\{{\bf W}^{(0)}\}$ and $\{{\bf \Upsilon}^{(0)}\}$.  Then, we employ AO and update either $\{{\bf W}\}$ or $\{{\bf \Upsilon}\}$ in each step. More particularly, we develop iterative algorithms in which we first update $\{{\bf W}\}$ while $\{{\bf \Upsilon}\}$ is fixed to $\{{\bf \Upsilon}^{(z-1)}\}$, where $z$ is the iteration index. Then we alternate and update $\{{\bf \Upsilon}\}$ when we fix $\{{\bf W}\}$ to $\{{\bf W}^{(z)}\}$. We iterate this operation until the algorithm converges. Below, we describe our detailed solutions for updating $\{{\bf W}\}$ and $\{{\bf \Upsilon}\}$. 

\subsection{Updating $\{{\bf W}\}$}
In this subsection, we solve \eqref{17} and \eqref{18} when $\{{\bf \Upsilon}\}$ is fixed to $\{{\bf \Upsilon}^{(z-1)}\}$. We first provide the solution for the max-min weighted rate and then solve \eqref{18}.
\subsubsection{Maximizing the minimum weighted rate} 
The max-min weighted rate for fixed RIS coefficients can be written as
\begin{subequations}\label{(19)}
\begin{align}
\label{(19a)}
 \underset{r,\{\mathbf{W}\},{\bf t}
 }{\max}\,\,\,  & 
  r \\
\label{(19b)}  \text{s.t.}\,\,   \,&  r_{lk}\!\!\left(\!\{{\bf W}\}\!,\!\{{\bf \Upsilon}^{(z-1}\}\!\!\right)\!
  \geq \!\max(\alpha_{lk}r,r_{lk}^{th}),\forall lk,
  \\ 
\label{(19c)}  & \text{Tr}(\mathbf{C}_l)\leq P_l,\,\, \forall l,
\\
\label{(19d)}
 &
t_{lk}\geq 0,\,\,\,\,\forall lk,
\\
\label{(19e)}
&
\sum_kt_{lk}\leq \min_k(r_{c,lk}),\,\,\forall l.
 \end{align}
\end{subequations}
Constraints \eqref{(19c)} and \eqref{(19d)} are convex, and \eqref{(19)} is linear in ${\bf t}$, and $r$. However, \eqref{(19)} is non-convex due to \eqref{(19b)} and \eqref{(19e)}. Specifically, $r_{p,lk}$ and $r_{c,lk}$ are non-concave functions of $\{{\bf W}\}$, which makes \eqref{(19)} non-convex. To derive a suboptimal solution for \eqref{(19)}, we leverage the MM method and calculate appropriate concave surrogate functions for $r_{p,lk}$ and $r_{c,lk}$. To this end, we utilize the lower bounds in \cite[Lemma 3]{soleymani2024optimization} and \cite[Lemma 4]{soleymani2024optimization}, which are provided below to facilitate the development of the algorithm.

\begin{lemma}[\!\cite{soleymani2022improper}]\label{lem-2}
Consider the arbitrary matrices ${\bf \Gamma}\in\mathbb{C}^{m\times n}$ and $\bar{\bf \Gamma}\in\mathbb{C}^{m\times n}$, and positive definite matrices ${\bf \Omega}\in\mathbb{C}^{m\times m}$ and $\bar{\bf \Omega}\in\mathbb{C}^{m\times m}$, where $m$ and $n$ are arbitrary natural numbers. Then, we have:
\begin{multline} 
\ln \left|\mathbf{I}+{\bf \Omega}^{-1}{\bf \Gamma}{\bf \Gamma}^H\right|\geq
 \ln \left|\mathbf{I}+{\bf \Omega}^{-1}\bar{{\bf \Gamma}}\bar{{\bf \Gamma}}^H\right|
\\-
\text{{\em Tr}}\left(
\bar{{\bf \Omega}}^{-1}
\bar{{\bf \Gamma}}\bar{{\bf \Gamma}}^H
\right)
+
2\mathfrak{R}\left\{\text{{\em Tr}}\left(
\bar{{\bf \Omega}}^{-1}
\bar{{\bf \Gamma}}{\bf \Gamma}^H
\right)\right\}\\
-
\text{{\em Tr}}\left(
(\bar{{\bf \Omega}}^{-1}-(\bar{{\bf \Gamma}}\bar{{\bf \Gamma}}^H + \bar{{\bf \Omega}})^{-1})^H({\bf \Gamma}{\bf \Gamma}^H+{\bf \Omega})
\right).
\label{lower-bound}
\end{multline}
\end{lemma} 
\begin{lemma}[\cite{soleymani2024optimization}]\label{lem-3} 
 The following inequality holds for arbitrary matrices ${\bf \Gamma}\in\mathbb{C}^{m\times n}$ and $\bar{{\bf \Gamma}}\in\mathbb{C}^{m\times n}$ and positive semi-definite ${\bf \Omega}\in\mathbb{C}^{m\times m}$ and $\bar{{\bf \Omega}}\in\mathbb{C}^{m\times m}$
\begin{multline}
\label{eq10}
f\left({\bf \Gamma},{\bf \Omega}\right)=\text{\em Tr}\left({\bf \Omega}^{-1}{\bf \Gamma}{\bf \Gamma}^H\right)\geq 
2\mathfrak{R}\left\{\text{\em Tr}\left(\bar{\bf \Omega}^{-1}\bar{\bf \Gamma}{\bf \Gamma}^H\right)\right\}
\\-
\text{\em Tr}\left(\bar{\bf \Omega}^{-1}\bar{\bf \Gamma}\bar{\bf \Gamma}^H\bar{\bf \Omega}^{-1}{\bf \Omega}\right),
\end{multline}
where $m$ and $n$ are arbitrary natural numbers, and $\mathfrak{R}\{x\}$ returns the real value of $x$.
\end{lemma}
The FBL rates $r_{p,lk}$ and $r_{c,lk}$ consist of two parts: the asymptotic Shannon rate and the term related to the channel dispersion. We utilize the quadratic, concave lower bound in Lemma \ref{lem-2} to calculate a surrogate function for the part related to the first-order Shannon rate of $r_{p,lk}$ and $r_{c,lk}$. Moreover, we use the results in Lemma \ref{lem-3} to compute a concave lower bound for the part related to the channel dispersion term. In the lemma below, we derive surrogate functions for $r_{p,lk}$ and $r_{c,lk}$, which are jointly quadratic and concave in the beamforming matrices $\{{\bf W}\}$.

\begin{lemma}\label{lem-4}
For all feasible $\{{\bf W}\}$, the following inequalities hold
\begin{multline}
\label{(22)}
r_{p,lk}\geq \tilde{r}_{p,lk}^{(z)} = a_{p,lk}
+2\sum_{ij}\mathfrak{R}\left\{\text{{\em Tr}}\left(
{\bf A}_{p,lk,ij}\mathbf{W}_{ij}^H
\bar{\mathbf{H}}_{lk,i}^H\right)\right\}
\\
+2\sum_{i\neq l}\mathfrak{R}\left\{\text{{\em Tr}}\left(
{\bf A}_{p,lk,i}\mathbf{W}_{i}^H
\bar{\mathbf{H}}_{lk,i}^H\right)\right\}
-
\text{{\em Tr}}\left(
{\bf B}_{p,lk}\mathbf{D}_{c,lk}
\right),
\end{multline}
\begin{multline}
\label{(23)}
r_{c,lk}\geq \tilde{r}_{c,lk}^{(z)} = a_{c,lk}
+2\sum_{ij}\mathfrak{R}\left\{\text{{\em Tr}}\left(
{\bf A}_{c,lk,ij}\mathbf{W}_{ij}^H
\bar{\mathbf{H}}_{lk,i}^H\right)\right\}
\\
+2\sum_{i}\mathfrak{R}\left\{\text{{\em Tr}}\left(
{\bf A}_{c,lk,i}\mathbf{W}_{i}^H
\bar{\mathbf{H}}_{lk,i}^H\right)\right\}
\\
-
\text{{\em Tr}}\left(
{\bf B}_{c,lk}(\bar{\mathbf{H}}_{lk,l}\mathbf{W}_{l}\mathbf{W}_{l}^H\bar{\mathbf{H}}_{lk,l}^H+\mathbf{D}_{c,lk})
\right),
\end{multline}
where we have:
\begin{align*}
a_{p,lk}&=\ln\left|{\bf I}+\bar{\bf D}_{lk}^{-1}\bar{\bf S}_{lk}\right|
-
\text{{\em Tr}}\left(
\bar{\bf D}_{lk}^{-1}\bar{\bf S}_{lk}
\right)\!
\nonumber\\
&\hspace{.2cm}
-\!\frac{Q^{-1}(\epsilon_p)}{2\sqrt{n}}\!\!
\left(\!\!\sqrt{\!2\text{\em Tr}\!\left(\!\bar{\bf S}_{lk}\bar{\bf D}_{c,lk}^{-1}\!\right)}\!+\!\frac{2I}{\sqrt{\!2\text{\em Tr}\!\left(\!\bar{\bf S}_{lk}\bar{\bf D}_{c,lk}^{-1}\!\right)}}\!
\right)\!\!,
\\
a_{c,lk}&=\ln\left|{\bf I}+\bar{\bf D}_{c,lk}^{-1}\bar{\bf S}_{c,lk}\right|
-
\text{{\em Tr}}\left(
\bar{\bf D}_{c,lk}^{-1}\bar{\bf S}_{c,lk}
\right)\!
\nonumber\\
&\hspace{1cm}
-\frac{Q^{-1}(\epsilon_c)}{2\sqrt{n}}
\left(\!\sqrt{
\bar{\zeta}_{c,lk}}
+\frac{2I}{\sqrt{\bar{\zeta}_{c,lk}}}
\right),
\\
{\bf A}_{p,lk,ij}\!&\!=\!\!\left\{\!\!\!\! \begin{array}{ll}
\bar{\mathbf{D}}^{-1}_{lk}\bar{\mathbf{H}}_{lk,l}\bar{\mathbf{W}}_{lk}
&\!
\text{\em if}
\,\,i=l,
\,j=k,\\
\frac{Q^{-1}(\epsilon_p)}{\sqrt{2n\text{\em Tr}\left(\bar{\bf S}_{lk}\bar{\bf D}_{c,lk}^{-1}\right)}}
\bar{\mathbf{D}}_{c,lk}^{-1}\bar{\mathbf{H}}_{lk,i}\bar{\mathbf{W}}_{ij}
&
\!
\text{\em otherwise},
\end{array}\right.
\\
{\bf A}_{p,lk,i}\!&\!=\frac{Q^{-1}(\epsilon_p)}{\sqrt{2n\text{\em Tr}\left(\bar{\bf S}_{lk}\bar{\bf D}_{c,lk}^{-1}\right)}}
\bar{\mathbf{D}}_{c,lk}^{-1}\bar{\mathbf{H}}_{lk,i}\bar{\mathbf{W}}_{i},
\\
{\bf A}_{c,lk,ij}\!&\!=\frac{Q^{-1}(\epsilon_c)}{\sqrt{n\bar{\zeta}_{c,lk}}}
(\bar{\mathbf{S}}_{c,lk}+\bar{\mathbf{D}}_{c,lk})^{-1}\bar{\mathbf{H}}_{lk,i}\bar{\mathbf{W}}_{ij},
\\
{\bf A}_{c,lk,i}\!&\!=\!\!\left\{\!\!\!\! \begin{array}{ll}
\bar{\mathbf{D}}^{-1}_{c,lk}\bar{\mathbf{H}}_{lk,l}\bar{\mathbf{W}}_{l}
&\!
\text{\em if}
\,\,i=l,\\
\frac{Q^{-1}(\epsilon_c)}{\sqrt{n\bar{\zeta}_{c,lk}}}
(\bar{\mathbf{S}}_{c,lk}+\bar{\mathbf{D}}_{c,lk})^{-1}\bar{\mathbf{H}}_{lk,i}\bar{\mathbf{W}}_{i},
&
\!
\text{\em otherwise},
\end{array}\right.
\\
{\bf B}_{p,lk}&\!=\bar{\mathbf{D}}^{-1}_{lk}\!\!-\bar{\mathbf{D}}^{-1}_{c,lk}\!\!
+\!
\frac{Q^{-1}(\epsilon_p)}{\sqrt{2n\text{\em Tr}\left(\bar{\bf S}_{lk}\bar{\bf D}_{c,lk}^{-1}\right)}}
\bar{\mathbf{D}}_{c,lk}^{-1}
\bar{\mathbf{D}}_{lk}\bar{\mathbf{D}}_{c,lk}^{-1},
\\
{\bf B}_{c,lk}&\!=\bar{\mathbf{D}}^{-1}_{c,lk}\!-\!(\bar{\mathbf{S}}_{c,lk}+ \bar{\mathbf{D}}_{c,lk})^{-1}\!\!
\\
&
+\!
\frac{Q^{-1}(\epsilon_c)}{\sqrt{n\bar{\zeta}_{c,lk}}}
(\bar{\mathbf{S}}_{c,lk}+ \bar{\mathbf{D}}_{c,lk})^{-1}
\bar{\mathbf{D}}_{c,lk}(\bar{\mathbf{S}}_{c,lk}+ \bar{\mathbf{D}}_{c,lk})^{-1}\!,
\end{align*}
where $\bar{\mathbf{D}}_{lk}$, $\bar{\mathbf{D}}_{c,lk}$, $\bar{\mathbf{S}}_{lk}$, $\bar{\mathbf{S}}_{c,lk}$, $\bar{\mathbf{W}}_{ij}$, $\bar{\mathbf{W}}_{i}$, and $\bar{\mathbf{H}}_{lk,i}$, $\forall l,k,i,j$ are, respectively, the initial values of ${\mathbf{D}}_{lk}$, ${\mathbf{D}}_{c,lk}$, ${\mathbf{S}}_{lk}$, ${\mathbf{S}}_{c,lk}$, ${\mathbf{W}}_{ij}$, ${\mathbf{W}}_{i}$, and ${\mathbf{H}}_{lk,i}$ at the current step, which are obtained upon replacing $\{\mathbf{W}\}$ by $\{\mathbf{W}^{(z-1)}\}$ and $\{\bf{\Upsilon}\}$ by $\{{\bf\Upsilon}^{(z-1)}\}$. Moreover, $I = \min(N_{BS}, N_u)$ and  $\bar{\zeta}_{c,lk}=2\text{\em Tr}\left(\bar{\mathbf{S}}_{c,lk}(\bar{\mathbf{D}}_{c,lk}+\bar{\mathbf{S}}_{c,lk})^{-1}\right)$ is the initial value of the channel dispersion for decoding the common message at U$_{lk}$.
\end{lemma} 
\begin{proof}
 Please refer to Appendix \ref{app-a}.
\end{proof}
We substitute $r_{p,lk}$ and $r_{c,lk}$ by $\tilde{r}_{p,lk}^{(z)}$ and $\tilde{r}_{c,lk}^{(z)}$  in \eqref{(19)}, respectively, which yields the convex problem
\begin{subequations}\label{(24)}
\begin{align}
\label{(24a)}
 \underset{r,\{\mathbf{W}\},{\bf t}
 }{\max}\,\,\,  & 
  r \\
\label{(24b)}  \text{s.t.}\,\,   \,&  \tilde{r}_{p,lk}^{(z)}+t_{lk}
  \geq \!\max(\alpha_{lk}r,r_{lk}^{th}),\forall lk,
\\
\label{(24c)}
&
\sum_kt_{lk}\leq \min_k(\tilde{r}_{c,lk}^{(z)}),\,\,\forall l,
\\ 
& 
\eqref{(19c)},\eqref{(19d)}.
 \end{align}
\end{subequations}
 {Since \eqref{(24)} is a convex optimization problem, the globally optimal solution of \eqref{(24)} can be efficiently computed by numerical optimization tools such as CVX \cite{grant2014cvx}.}

\subsubsection{Maximizing the minimum weighted EE}
 The max-min weighted EE optimization problem falls into the category of fractional programming problems, since the EE terms are fractional functions of $\{{\bf W}\}$. To solve \eqref{18} for fixed $\{{\bf \Upsilon}^{(z)}\}$, we replace $r_{p,lk}$ and $r_{c,lk}$ with $\tilde{r}_{p,lk}^{(z)}$ and $\tilde{r}_{c,lk}^{(z)}$, respectively, in \eqref{18}, which results in
\begin{subequations}\label{(25)}
\begin{align}
\label{(25a)}
 \underset{e,\{\mathbf{W}\},{\bf t}
 }{\max}\,\,\,  & 
  e \\
\label{(25b)}  \text{s.t.}\,\,   \,&  \tilde{e}_{lk}^{(z)}
  \geq \lambda_{lk}e,\,\,\forall lk,
\\
\label{(25c)}
&
\tilde{r}_{p,lk}^{(z)}+t_{lk}\geq r_{lk}^{th},\forall lk,
\\ 
& 
\eqref{(19c)},\eqref{(19d)},\eqref{(24c)},
 \end{align}
\end{subequations}
where
\begin{equation}\label{eq-e-tilde}
    \tilde{e}_{lk}^{(z)}=
    \frac{\tilde{r}_{lk}^{(z)}+t_{lk}}{p_c+\eta\text{Tr}\left({\bf W}_{lk}{\bf W}_{lk}^H\right)
     +\frac{\eta}{k} \text{Tr}\left({\bf W}_{l}{\bf W}_{l}^H\right)
     }.
\end{equation}
The constraints \eqref{(19c)}, \eqref{(19d)}, \eqref{(24c)}, and \eqref{(25c)} are convex, but \eqref{(25)} is non-convex since $\tilde{e}_{lk}$ has a fractional structure in $\{{\bf W}\}$, and it is non-concave. Since the numerator of $\tilde{e}_{lk}$ is concave for all $l, k$, and its denominator is convex, we can calculate the optimum of \eqref{(25)} by the GDA. More specifically, the optimum of \eqref{(25)} can be found by iteratively solving
\begin{subequations}\label{(27)}
\begin{align}
\label{(27a)}
 \underset{e,\{\mathbf{W}\},{\bf t}
 }{\max}\,\,\,  & 
  e \\
\label{(27b)}  \text{s.t.}\,\,   \,&  \tilde{r}_{lk}^{(z)}\!+t_{lk}-\!\mu^{(z,m)}p_{lk}\left(\{{\bf W}\}\right)
  \!\geq\! \lambda_{lk}e,\,\,\forall lk,\!
\\ 
& 
\eqref{(19c)},\eqref{(19d)},\eqref{(24c)},\eqref{(25c)},
 \end{align}
\end{subequations}
and updating the constant coefficient $\mu^{(z,m)}$ as
\begin{equation}\label{eq-mu-28}
    \mu^{(z,m)}=\min_{lk}\left\{\frac{\tilde{r}_{lk}^{(z)}\left(\{{\bf W}^{(m-1)}\}\right)}{p_{lk}\left(\{{\bf W}^{(m-1)}\}\right)}\right\},
\end{equation}
where $m$ is the iteration index of the GDA, and
\begin{equation}\label{(29)}
   p_{lk}\left(\{{\bf W}\}\right)= p_c+\eta\text{Tr}\left({\bf W}_{lk}{\bf W}_{lk}^H\right)
     +\frac{\eta}{k} \text{Tr}\left({\bf W}_{l}{\bf W}_{l}^H\right).
\end{equation}

\subsubsection{Discussions on single-stream RSMA}
Utilizing MISO or SIMO systems restricts the maximum number of streams per user to $1$. An advantage of MIMO systems is the ability to support multiple data streams per user, significantly enhancing both the diversity gain and overall system performance. Of course, these benefits can be achieved only when the system is optimized for multiple-stream data transmissions. However, one can also use the multiple transmit and receive antennas for beamforming only, thereby improving the SINR, utilizing single-stream data transmission. To implement single-stream RSMA, one can replace matrices ${\bf W}_l $ and ${\bf W}_{lk} $ with vectors ${\bf w}_l\in\mathbb{C}^{N_{BS} \times 1} $ and ${\bf w}_{lk}\in\mathbb{C}^{N_{BS} \times 1} $, respectively. Note that MISO systems are limited to single-stream data transmission. Still, the RSMA algorithms harnessed for MISO systems like the ones in \cite{soleymani2023optimization} cannot be applied to  MIMO systems since they are designed for single-antenna users. Indeed, the SE and EE optimization of MIMO systems is much more complex than in the MISO ones, even when single-stream data transmission is employed. To elaborate on this issue, we restate $r_{c,lk}$ for the case that the common message transmitted by BS $l$ is single stream:
\begin{multline}
    r_{c,lk,\text{Single-Stream}}=
\log \left({\bf I} +{\bf w}_l^H{\bf H}_{lk,l}^H{\bf D}^{-1}_{c,lk}{\bf H}_{lk,l}{\bf w}_l \right)
\\
-Q^{-1}(\epsilon_c)\sqrt{\frac{2{\bf w}_l^H{\bf H}_{lk,l}^H({\bf D}_{c,lk}+{\bf S}_{c,lk})^{-1}{\bf H}_{lk,l}{\bf w}_l}{n}},
\end{multline}
where ${\bf D}_{c,lk}$ and ${\bf S}_{c,lk}$ are defined as in \eqref{eq-9}.
To simplify $r_{c,lk}$ for the single-stream data transmission, we use the Sylvester determinant identity, stating that $|{\bf I}+ {\bf A}{\bf B} |=|{\bf I}+ {\bf B}{\bf A} |$, where ${\bf A}$ and ${\bf B}$ are arbitrary matrices. Note that  ${\bf D}_{c,lk}$ and $({\bf D}_{c,lk}+{\bf S}_{c,lk})^{-1}$ are not rank one even when all transmitted signals are single stream. Indeed, the rate expressions in MIMO systems depend on the covariance matrix of the interference plus noise, even when single-stream data transmission is utilized. On the contrary, the covariance of the interference plus noise is scalar in MISO systems, which makes reusing MISO designs impossible as single-stream MIMO schemes.

\subsection{Updating $\{\bf{\Upsilon}\}$}
In this subsection, we solve \eqref{17} and \eqref{18} when $\{{\bf W}\}$ is fixed to $\{{\bf W}^{(z)}\}$. In this case, the solution of \eqref{17} can be
derived similarly to the solution of \eqref{18}. Indeed, when we fix $\{{\bf W}\}$, the EE of a user becomes a scaled function of the rate, and it is not a fractional function in the optimization variables $\{{\bf \Upsilon}\}$ and ${\bf t}$. Therefore, we provide only the solution of \eqref{18}, which we can formulate as a max-min weighted rate problem in $\{{\bf \Upsilon}\}$ for fixed $\{{\bf W}^{(z)}\}$ as follows
\begin{subequations}\label{(30)}
\begin{align}
\label{(30a)}
 \underset{e,\{{\bf \Upsilon}\}\in\mathcal{F},{\bf t}
 }{\max}\,\,\,  & 
  e \\
\label{(30b)}  \text{s.t.}\,\,   \,&  {r}_{lk}
  \geq \max(\lambda_{lk}p_{lk}^{(z)}e,r_{lk}^{th}),\,\,\forall lk,
  \\
  & \eqref{17d},\eqref{17e},
 \end{align}
\end{subequations}
where $r_{lk} = r_{p,lk} + t_{lk}$ and  $p_{lk}^{(z)}=p_{lk}\left(\{{\bf W}^{(z)}\}\right)$ is given by inserting $\{{\bf W}^{(z)}\}$ into \eqref{(29)}. In the following, we solve \eqref{(30)} for $\mathcal{F}_U$ and $\mathcal{F}_I$ in separate parts.

\subsubsection{Updating $\{\bf{\Upsilon}\}$ for $\mathcal{F}_U$ }
Since the rates $r_{p,lk}$ and $r_{c,lk}$ are not concave in $\{{\bf \Upsilon}\}$, \eqref{(30)} is non-convex. To solve \eqref{(30)}, we first calculate appropriate surrogate functions for the rates as in Lemma \ref{lem-4}, since the structures of the rates are similar in $\{\bf{\Upsilon}\}$ and $\{{\bf W}\}$.
\begin{corollary}\label{cor-1}
For all feasible $\{\bf{\Upsilon}\}$ and fixed $\{{\bf W}\}$, the following inequalities hold
\begin{multline}
\label{(31)}
r_{p,lk}\geq \hat{r}_{p,lk}^{(z)} = a_{p,lk}
+2\sum_{ij}\mathfrak{R}\left\{\text{{\em Tr}}\left(
{\bf A}_{p,lk,ij}\bar{\mathbf{W}}_{ij}^H
{\mathbf{H}}_{lk,i}^H\right)\right\}
\\
+\!2\sum_{i\neq l}\mathfrak{R}\left\{\text{{\em Tr}}\left(
{\bf A}_{p,lk,i}\bar{\mathbf{W}}_{i}^H
{\mathbf{H}}_{lk,i}^H\right)\right\}\!
-
\text{{\em Tr}}\left(
{\bf B}_{p,lk}\mathbf{D}_{c,lk}
\right),
\end{multline}
\begin{multline}
\label{(32)}
r_{c,lk}\geq \hat{r}_{c,lk}^{(z)} = a_{c,lk}
+2\sum_{ij}\mathfrak{R}\left\{\text{{\em Tr}}\left(
{\bf A}_{c,lk,ij}\bar{\mathbf{W}}_{ij}^H
{\mathbf{H}}_{lk,i}^H\right)\right\}
\\
+2\sum_{i}\mathfrak{R}\left\{\text{{\em Tr}}\left(
{\bf A}_{c,lk,i}\bar{\mathbf{W}}_{i}^H
{\mathbf{H}}_{lk,i}^H\right)\right\}
\\
-
\text{{\em Tr}}\left(
{\bf B}_{c,lk}(\mathbf{H}_{lk,l}\bar{\mathbf{W}}_{l}\bar{\mathbf{W}}_{l}^H\mathbf{H}_{lk,l}^H+\mathbf{D}_{c,lk})
\right),
\end{multline}
where the coefficients $a_{p,lk}$, $a_{c,lk}$, ${\bf A}_{p,lk,ij}$, ${\bf A}_{c,lk,ij}$, ${\bf A}_{p,lk,i}$, ${\bf A}_{c,lk,i}$, ${\bf B}_{p,lk}$, and  ${\bf B}_{c,lk}$ are defined as in Lemma \ref{lem-4}, and $\bar{\bf W}_{l}={\bf W}_{l}^{(z)}$ and $\bar{\bf W}_{lk}={\bf W}_{lk}^{(z)}$ for all $l,k$.
\end{corollary}
\begin{proof}
    The rates have a similar structure in $\{\bf{\Upsilon}\}$ and $\{{\bf W}\}$. Thus, this corollary can be proved similar to Appendix \ref{app-a}. Specifically, we can employ the bound in Lemma \ref{lem-2} to calculate a concave lower bound for the part of the FBL rate, related to the Shannon rate. Additionally, we can employ the bounds in \eqref{(38)} and Lemma \ref{lem-3} to derive a suitable lower bound for the part of the FBL rate, related to the channel dispersion, packet length and the reliability constraint.
\end{proof}
If we insert $\hat{r}_{p,lk}^{(z)}$ and $\hat{r}_{c,lk}^{(z)}$ into \eqref{(30)} for $\mathcal{F}_U$, we have
\begin{subequations}\label{(33)}
\begin{align}
\label{(33a)}
 \underset{e,\{{\bf \Upsilon}\},{\bf t}
 }{\max}\,\,\,  & 
  e \\
\label{(33b)}  \text{s.t.}\,\,   \,&  \hat{r}_{p,lk}^{(z)}+{t}_{lk}
  \geq \max(\lambda_{lk}p_{lk}^{(z)}e,r_{lk}^{th}),\,\,\forall lk,
  \\
  \label{(33c)} 
  & \sum_kt_{lk}\leq \min_k (\hat{r}_{c,lk}^{(z)}),\,\, \forall l,
  \\
  &
  \label{(33d)}
  t_{lk} \geq 0,  \,\,\, \forall l,k,
  \\ &
  \label{(33e)} |\upsilon_{mn}|^2 \leq 1, \,\,\, \forall m,n,
 \end{align}
\end{subequations}
which is convex and can be efficiently solved.  {For $\mathcal{F}_U$, the framework converges to a \textit{stationary point} of \eqref{18} (or \eqref{17}), since the concave lower bounds in Lemma \ref{lem-4} and Corollary \ref{cor-1} fulfill the three conditions mentioned in \cite[Section III]{soleymani2020improper}. We summarize our solution for \eqref{18} with $\mathcal{F}_U$ in Algorithm I.} 

\doublespacing 
\begin{table}[htb]
\small
\begin{tabular}{l}
\hline 
  {\textbf{Algorithm I} Max-min weighted EE with $\mathcal{F}_U$.} \\
\hline
\hspace{0.2cm}{\textbf{Initialization}}\\
\hspace{0.2cm}{Set $\gamma_1$, $\gamma_2$, $z=1$,  $\{\mathbf{W}\}=\{\mathbf{W}^{(0)}\}$, and$\{{\bf \Upsilon}\}=\{{\bf \Upsilon}^{(0)}\}$ }\\
\hline 
\hspace{0.2cm}
{\textbf{While} $\left(\underset{\forall lk}{\min}\,e^{(z)}_{lk}-\underset{\forall lk}{\min}\,e^{(z-1)}_{lk}\right)/\underset{\forall lk}{\min}\,e^{(z-1)}_{lk}\geq\gamma_1$ }\\ 
\hspace{.6cm}{{\bf Update} $\{\mathbf{W}\}$ {\bf for fixed} $\{{\bf \Upsilon}^{(z-1)}\}$}\\
\hspace{1.2cm}{Calculate $\tilde{r}_{lk}$ 
according to Lemma \ref{lem-4}}\\ 
\hspace{1.2cm}{Calculate $\tilde{e}_{lk}$ 
based on \eqref{eq-e-tilde}}\\ 
\hspace{1.2cm}{Compute $\{{\bf W}\}$ by solving \eqref{(25)}, i.e., by running}\\
\hspace{1.2cm}{\textbf{While} $\left(\underset{\forall lk}{\min}\,\tilde{e}^{(m)}_{lk}-\underset{\forall lk}{\min}\,\tilde{e}^{(m-1)}_{lk}\right)/\underset{\forall lk}{\min}\,\tilde{e}^{(m-1)}_{lk}\geq\gamma_2$ }\\ 
\hspace{1.8cm}{Update $\mu^{(z,m)}$ based on \eqref{eq-mu-28}}\\
\hspace{1.8cm}{Update $\{{\bf W}\}$ by solving \eqref{(27)}}\\
\hspace{.6cm}{{\bf Update} $\{{\bf\Upsilon}\}$ {\bf for fixed} $\{\mathbf{W}^{(z-1)}\}$}\\
\hspace{1.2cm}{Derive $\hat{r}_{lk}^{(z)}$  
according to Corollary 1}\\
\hspace{1.2cm}{Calculate $\{{\bf \Upsilon}^{(z)}\}$ by solving \eqref{(33)}}\\
\hspace{.6cm}{$z=z+1$}\\
\hspace{0.2cm}{\textbf{End (While)}}\\
\hspace{0.2cm}{{\bf Return} $\{\mathbf{W}^{(\star)}\}$ and $\{{\bf \Upsilon}^{(\star)}\}$.}\\
\hline 
\end{tabular} 
\end{table}
\singlespacing

\subsubsection{Updating $\{\bf{\Upsilon}\}$ for $\mathcal{F}_I$}
To update $\{\bf{\Upsilon}\}$ for $\mathcal{F}_I$, we have an additional non-convex constraint $|\upsilon_{mn}|^2\geq 1$ for all $m$ and $n$, which can be approximated as a convex constraint by using the convex-concave procedure similar to \cite[Eq. (43)]{soleymani2024optimization} as
\begin{align}\label{(34)}
|\upsilon_{mn}|^2\!\!\geq\!|\upsilon_{mn}^{(z-1)}|^2\!\!-\!2\mathfrak{R}\!\{\upsilon_{mn}^{(z-1)^*}\!(\upsilon_{mn}\!\!-\!\upsilon_{mn}^{(z-1)})\!\}\!\!\geq \!1\!-\!\delta,
\end{align} 
for all $m,n$, where $\delta>0$. In this case, we have
\begin{subequations}\label{(35)}
\begin{align}
\label{(35a)}
 \underset{e,\{{\bf \Upsilon}\},{\bf t}
 }{\max}\,\,\,  & 
  e \\
\label{(35b)}  \text{s.t.}\,\,   \,&  \!|\upsilon_{mn}^{(z-1)}|^2\!\!\!-\!2\mathfrak{R}\!\{\!\upsilon_{mn}^{(z-1)^*}\!(\upsilon_{mn}\!\!-\!\!\upsilon_{mn}^{(z-1)})\!\}\!\!\geq\!\! 1\!\!-\!\delta,\forall m,\!n,
  \\ &
  \eqref{(33b)},\eqref{(33c)},\eqref{(33d)},\eqref{(33e)}.
 \end{align}
\end{subequations}
The problem \eqref{(35)} is convex, but its solution, denoted by $\{{\bf \Upsilon}^{(\star)}\}$, may not  satisfy $|\upsilon_{mn}|=1$ for all $m,n$ because of the relaxation in \eqref{(34)}.  {Thus, we normalize $\{{\bf \Upsilon}^{(\star)}\}$ as $\hat{\upsilon}_{mn}=\frac{{\upsilon}_{mn}^{(\star)}}{|{\upsilon}_{mn}^{(\star)}|}$ for all $m,n$, and update $\{{\bf \Upsilon}\}$ for $\mathcal{F}_I$ according to
 \begin{equation}\label{(36)}
\{{\bf \Upsilon}^{(z)}\}=
\left\{
\begin{array}{lcl}
\{\hat{{\bf\Upsilon}}\}&\text{if}&
f_0\left(\left\{\mathbf{W}^{(z)}\right\},\{\hat{{\bf\Upsilon}}\}\right)\geq
\\
&&
f_0\left(\left\{\mathbf{W}^{(z)}\right\},\{{\bf\Upsilon}^{(z-1)}\}\right)
\\
\{{\bf\Upsilon}^{(z-1)}\}&&\text{otherwise},
\end{array}
\right.
\end{equation}
to ensure convergence by generating a \textit{non-decreasing} sequence of $f_0\left(\left\{\mathbf{W}\right\},\{{{\bf\Upsilon}}\}\right)$, where $f_0=\min_{lk}\{\frac{r_{lk}}{\alpha_{lk}}\}$ and $f_0=\min_{lk}\{\frac{e_{lk}}{\lambda_{lk}}\}$ when solving \eqref{17} and \eqref{18}, respectively.}

\subsection{Computational Complexity Analysis}
In this subsection, we estimate an approximate upper bound for the number of multiplications required by our algorithms to obtain a solution for \eqref{18} with the feasibility set $\mathcal{F}_U$. The computational complexity of solving \eqref{17} can be similarly approximated. In the following, we first approximate the number of multiplications to update $\{{\bf W}\}$ and $\{{{\bf\Upsilon}}\}$ in separate paragraphs.  

 As shown in Algorithm I, updating $\{{\bf W}\}$ requires solving the convex optimization problem in \eqref{eq-e-tilde}. To numerically solve such a problem,  the number of Newton iterations increases proportionally with the square root of the number of its constraints \cite[Chapter 11]{boyd2004convex}, which is equal to $L(3K+2)$ in \eqref{eq-e-tilde}. Solving each Newton iteration needs to calculate  $\tilde{r}_{c,lk}$ and $\tilde{r}_{p,lk}$ for all users. The number of multiplications to compute $\tilde{r}_{p,lk}$ is approximately in the same order as that of $\tilde{r}_{c,lk}$. These surrogate rates are quadratic in $\{{\bf W}\}$, and to compute each  approximately requires $\mathcal{O}\left[LKN_{BS}^2(2N_{BS}+N_{u})\right]$ multiplications. As we employ the GDA, \eqref{eq-e-tilde} has to be iterated to obtain the solution of \eqref{(25)}. We set the maximum number of iterations of the GDA to $J$. Thus, the total number of multiplications to update $\{{\bf W}\}$ at each iteration of our algorithm can be approximated as $\mathcal{O}\left[J L^2K^2N_{BS}^2\sqrt{L(3K+2)}(2N_{BS}+N_{u})\right]$.

 To update $\{{\bf \Upsilon}\}$ requires solving  the convex optimization problem in \eqref{(33)}, which has $L(2K+1)+MN_{RIS}$ constraints. To solve each Newton iteration, one needs to calculate $2LK$ surrogate rates and $L^2K$ equivalent channels. The number of multiplications to compute $\hat{r}_{p,lk}$ (or $\hat{r}_{c,lk}$) in Corollary \ref{cor-1}  is on the same order of that of $\tilde{r}_{p,lk}$ (or $\tilde{r}_{p,lk}$), which is approximately $\mathcal{O}\left[LKN_{BS}^2(2N_{BS}+N_{u})\right]$. 
To calculate each channel, ${\bf H}_{lk,i}$, $\forall l,k,i$, approximately needs $MN_uN_{BS}N_{RIS}$ multiplications. Therefore, to update $\{{\bf \Upsilon}\}$ approximately needs $\mathcal{O}[L^2K\sqrt{L(2K+1)+MN_{RIS}}(KN_{BS}^2(2N_{BS}+N_{u})+MN_uN_{BS}N_{RIS})]$ multiplications.

\section{Numerical results}\label{sec-iv}
In this section, we provide numerical results through Monte Carlo simulations. To this end, we consider a two-cell BC with two multiple-antenna BSs, each serving $K$ multiple-antenna users, and one RIS in each cell as shown in \cite[Fig. 2]{soleymani2022rate}. We assume that there is a line of sight (LoS) link between each RIS and all users as well as all the BSs. Specifically, the small-scale fading of the RIS-related links follows a Rician distribution with a Rician factor of $3$. Moreover, we assume that the link between each BS and each user is non-LoS (NLoS), and follows a Rayleigh distribution. Additionally, we assume that all the users have the same priority, which means that all the users have the same weight, i.e., $\alpha_{lk} = 1$ (or $\lambda_{lk}=1$) for all $l$, $k$. In other words, we consider the maximization of the minimum rate (or EE), which we refer to as the max-min rate (or EE). Furthermore, the power budget of the BSs is set to $P$. We also assume that $\epsilon_c = \epsilon_p = \epsilon/2$, where $\epsilon$ is the maximum acceptable bit error rate, and $n$ is the packet length. The other parameters, including the large-scale fading, are selected based on \cite{soleymani2022improper}. 

To make the numerical results consistent, we define a scenario for the max-min rate numerical results, referred to as "{\bf Scenario 1}" in which $N_{{BS}}=2$, $N_{u}=2$, $P=10$ dB, $K=3$, $L=2$,  $M=2$, $\epsilon=10^{-5}$,  $n=256$ bits, and $N_{{RIS}}=20$. Additionally, we define "{\bf Scenario 2}" for the max-min EE numerical results in which $N_{BS} = 2$, $N_u = 2$, $K = 3$,  $P=10$ dB, $L = 1$, $M = 1$, $\epsilon=10^{-5}$, $n=256$ bits, and $N_{RIS} = 20$. In the numerical results, we mainly consider the parameters from {\bf Scenario 1}/{\bf Scenario 2} for the SE/EE results, unless we explicitly mention otherwise. Obviously, if we refer to {\bf Scenario 1} while considering the impact of a specific parameter (e.g., packet length), all the other parameters are chosen based on the values in {\bf Scenario 1}, except the specific parameter (i.e., packet length), which varies as depicted in the corresponding figure.

Note that to the best of our knowledge, there is no work on RSMA and/or NOMA in MU-MIMO RIS-aided systems with FBL coding and supporting multiple-stream data transmission per user. Thus, we compare the performance of the proposed RSMA schemes to the scheme in \cite{soleymani2024optimization}, which uses TIN, and to the NOMA-based schemes in \cite{soleymani2022noma, soleymani2023noma}, which is based on the Shannon rate. The legends in the figures are defined as:
\begin{itemize}
\item {\bf RIS-RS}/{\bf RIS-RS$_I$} (or {\bf RIS-TIN}):  The RSMA algorithms (or the TIN algorithms in \cite{soleymani2024optimization}) conceived for MIMO RIS-aided URLLC system with $\mathcal{F}_U$/$\mathcal{F}_I$ (or $\mathcal{F}_U$).

\item {\bf No-RIS-RS} (or {\bf No-RIS-TIN}):  The RSMA design (or the TIN algorithms in \cite{soleymani2024optimization}) for MIMO URLLC systems that do not use RIS.

\item {\bf RIS-Rand-RS} (or {\bf RIS-Rand-TIN}):  The RSMA scheme (or the scheme with TIN in \cite{soleymani2024optimization}) for MIMO RIS-aided URLLC systems that utilize random RIS coefficients.

\item {\bf  RIS-Sh}:  The RSMA algorithms for MIMO RIS-aided systems with the asymptotic Shannon rate and $\mathcal{F}_U$.

\item  {{\bf  SS-RIS-RS} (or {\bf  SS-RIS-TIN}):  The single-stream RSMA algorithms (or the single-stream TIN algorithms in \cite{soleymani2024optimization}) conceived for MIMO RIS-aided URLLC system with $\mathcal{F}_U$.}

\item {\bf  NOMA-Sh} (or {\bf  NOMA}):  The NOMA-based scheme proposed in \cite{soleymani2022noma, soleymani2023noma} employing proper Gaussian signaling and
perfect devices for MIMO RIS-aided systems with the Shannon (or FBL) rate and $\mathcal{F}_U$.
\end{itemize}

\subsection{Maximization of the Minimum Rate}\label{sec-iv-a}
We evaluate the max-min rate of the proposed RSMA schemes in this subsection. We also examine how RSMA can influence the performance of RISs and vice versa. To this end, we investigate the impact of different system parameters like the BSs’ power budget, the packet length, the reliability constraint, and the number of users per cell.

\begin{figure}[t]
    \centering
       \includegraphics[width=.44\textwidth]{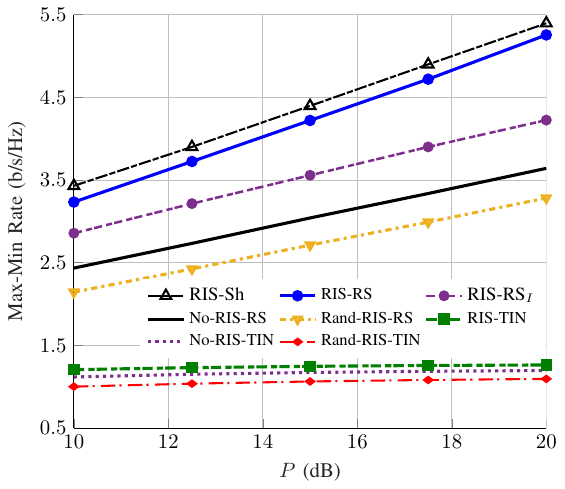}
    \caption{Average max-min rate versus power budget at the BSs for {\bf Scenario 1}. 
    }
	\label{Fig-rr1} 
\end{figure}

\subsubsection{Impact of power budget}
Fig. \ref{Fig-rr1} depicts the average max-min rate versus the power budget at the BSs for {\bf Scenario 1}. This may not be considered as a system of high user load, since the total number of antennas is higher than the number of users. However, the number of BS antennas is smaller than the number of users per cell. In this figure, RSMA substantially increases the average max-min rate. Indeed, the RSMA scheme for systems operating without RIS substantially outperforms the TIN scheme proposed in \cite{soleymani2024optimization} for the RIS-aided systems. Additionally, RSMA with RIS significantly outperforms the other schemes, when the RIS elements are optimized based on our proposed scheme.

\begin{figure}
    \centering
    \begin{subfigure}
{0.24\textwidth}
        \centering
           \includegraphics[width=\textwidth]{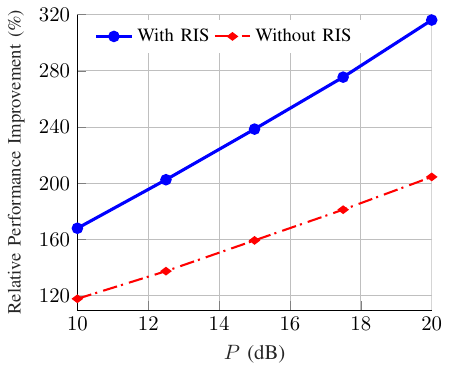}
        \caption{Benefits of RSMA over
SDMA.}
    \end{subfigure}
\begin{subfigure}
{0.24\textwidth}
        \centering
       \includegraphics[width=\textwidth]{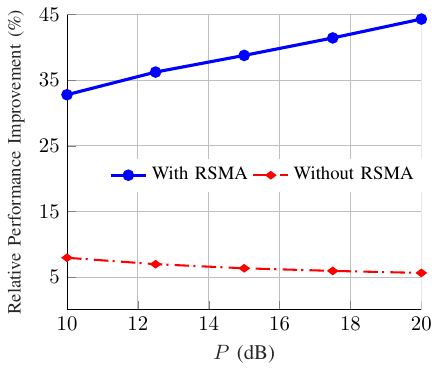}
        \caption{Benefits of RIS over systems
operating without RISs.}
    \end{subfigure}%
    \caption{Average max-min rate gains of utilizing RSMA and RISs
versus the power budget at the BSs for {\bf Scenario 1}. }
	\label{Fig-rr2} 
\end{figure}
Fig. \ref{Fig-rr2} shows the average max-min rate gains of utilizing RSMA and RISs versus the power budget at the BSs for {\bf Scenario 1}. The benefits of RSMA are calculated as the percentage of improvements in the max-min rate over the SDMA schemes, which utilize TIN. Specifically, the RSMA benefits associated with (without) RIS are derived by comparing the average max-min rate of the RIS-RS (No-RIS-RS) scheme with the RIS-TIN (No-RIS-TIN) scheme. Moreover, the benefits of RIS with RSMA (or TIN) are obtained by comparing the max-min rate of the RIS-RS (RIS-TIN) scheme to that of the No-RIS-RS (No-RIS-TIN) scheme. Interestingly, RISs boost the RSMA gains in this example. Moreover, the gains of RSMA scale with the power budget at the BSs. This happens, since increasing the power budget may not lead to a significant rate increment, when the system is interference limited. However, when a powerful interference-management technique is used, the interference is mitigated, and hence, the interference limits performance to a lesser extent especially in the high SNR regime. This is also in line with previous results in \cite{soleymani2022improper, soleymani2019improper, soleymani2020improper}, where it was shown that the benefits of IGS increase with the transmitters' power budget $P$.

Another interesting observation in Fig. \ref{Fig-rr2} is that RSMA may augment the benefits of RISs. In this example, the gain with RISs for TIN is almost negligible (less than 10\%). However, the RIS gain significantly increases when RSMA is employed. Additionally, the RIS gain decreases with the BS power budget when TIN is used. This is in stark contrast to the case when RSMA is employed, as described before. Note that, in this figure, $N_{RIS} = 20$, which is relatively low. It is known that the RIS gains become more significant when $N_{RIS}$ grows.

\subsubsection{Impact of the packet length}

\begin{figure}[t]
    \centering
\includegraphics[width=.44\textwidth]{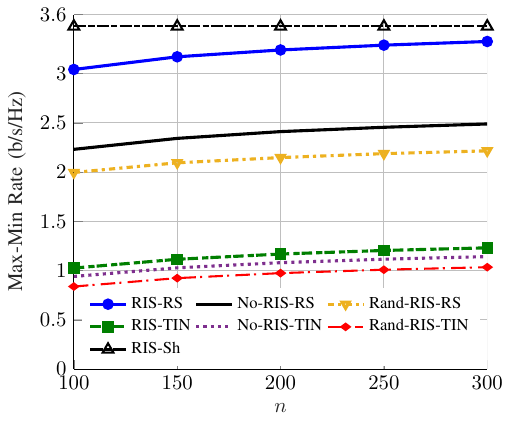}
    \caption{Average max-min rate versus the packet length for {\bf Scenario 1}. 
    }
	\label{Fig-rr3} 
\end{figure} 
Fig. \ref{Fig-rr3} illustrates the average max-min rate versus the packet length for {\bf Scenario 1}. In this example, RSMA for RIS-aided systems with optimized RIS elements substantially outperforms all the other algorithms. Moreover, RSMA provides much more significant gains compared to RISs, and all the RSMA schemes outperform the TIN schemes both with and without an RIS. Additionally, the max-min rate grows with $n$, and it is expected to converge to the Shannon rate in the asymptotic case.

\begin{figure}[t]
    \centering
    \begin{subfigure}[t]{0.24\textwidth}
        \centering
           \includegraphics[width=\textwidth]{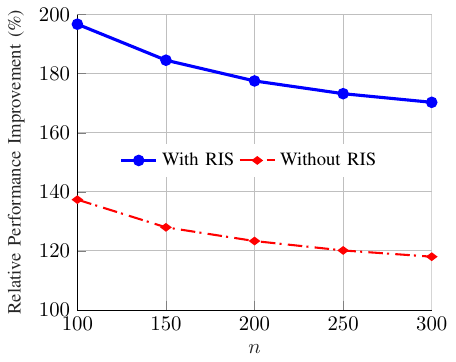}
        \caption{Benefits of RSMA over SDMA.}
    \end{subfigure}
\begin{subfigure}[t]{0.24\textwidth}
        \centering
       \includegraphics[width=\textwidth]{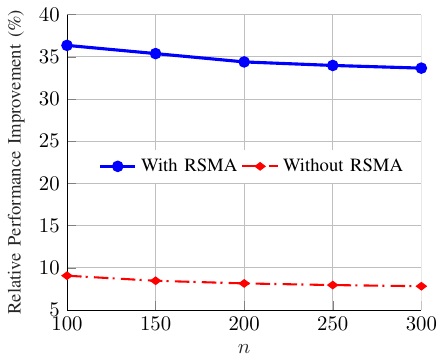}
        \caption{Benefits of RIS over systems
operating without RISs.}
    \end{subfigure}%
    \caption{Average max-min rate benefits of utilizing RSMA and RISs
versus $n$ for {\bf Scenario 1}. }
	\label{Fig-rr4} 
\end{figure}
Fig. \ref{Fig-rr4} shows the average max-min rate gains of utilizing RSMA and RISs versus $n$ for {\bf Scenario 1}. In this example, the RISs significantly amplify the RSMA gains. Moreover, the RSMA gains are much more significant when the packets are shorter. As described in Section \ref{sec-ii-d}, the latency constraint is related to the packet length, and shorter packets have to be utilized when a lower latency has to be achieved. As a result, RSMA is more beneficial in URLLC systems. We also observe that the RIS gains are less than 10\% when TIN is used, but the gains considerably increase when RSMA is employed. Specifically, RSMA enhances the gains of RISs by more than a factor of $3$ in this example. Furthermore, the RIS gains are reduced when $n$ grows, which shows that RISs provide higher gains in URLLC systems.

\subsubsection{Impact of  the reliability constraint}

\begin{figure}[t]
    \centering
\includegraphics[width=.44\textwidth]{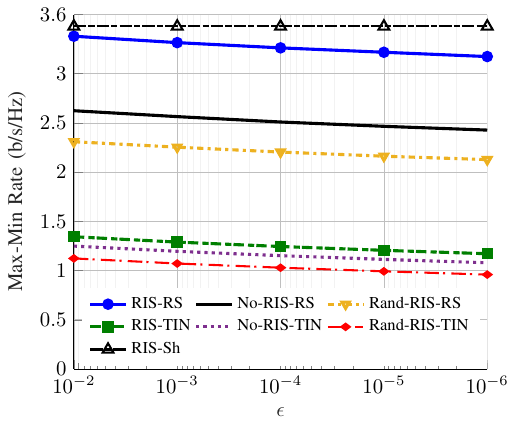}
    \caption{Average max-min rate versus $\epsilon$ for {\bf Scenario 1}.
    }
	\label{Fig-rr5} 
\end{figure} 
Fig. \ref{Fig-rr5} depicts the average max-min rate versus the maximum tolerable bit error rate for {\bf Scenario 1}. Again, RSMA substantially outperforms the TIN schemes both with and without RISs. Moreover, the proposed RSMA scheme using RISs significantly outperforms the other schemes, if the RIS elements are optimized by the proposed algorithms. Additionally, the max-min rate increases with $\epsilon$. This means that the lower the maximum tolerable bit error rate is, the lower the data transmission rate has to be.

\begin{figure}[t]
    \centering
    \begin{subfigure}[t]{0.23\textwidth}
        \centering
           \includegraphics[width=\textwidth]{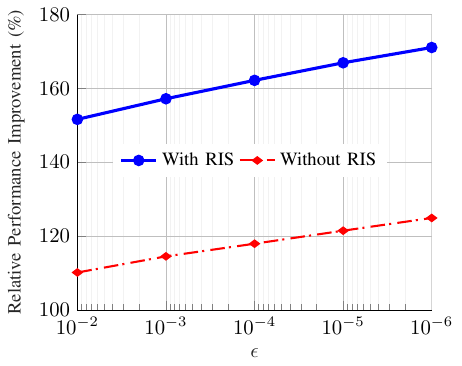}
        \caption{Benefits of RSMA over SDMA.}
    \end{subfigure}
~~
\begin{subfigure}[t]{0.23\textwidth}
        \centering
       \includegraphics[width=\textwidth]{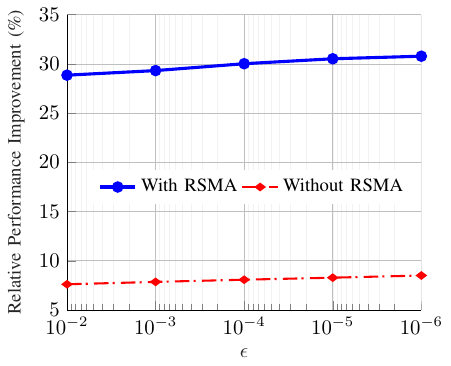}
        \caption{Benefits of RIS over systems operating without RISs.}
    \end{subfigure}%
    \caption{Average max-min rate benefits of utilizing RSMA and RISs
versus $\epsilon$ for {\bf Scenario 1}. }
	\label{Fig-rr6} 
\end{figure}
Fig. \ref{Fig-rr6} shows the average max-min rate benefits of utilizing RSMA and RISs versus $n$ for {\bf Scenario 1}. The benefits of RSMA increase with $\epsilon^{-1}$. Thus, the lower the maximum tolerable bit error rate is, the higher improvements along with RSMA. Moreover, RISs noticeably improve the gains of RSMA. Additionally, we can observe that RSMA drastically increases the gains along with RISs. Furthermore, in this example, the gains associated with RISs slightly increase, when $\epsilon$ is reduced.

\subsubsection{Impact of the number of users per cell}
\begin{figure}[t]
    \centering
      \includegraphics[width=.34\textwidth]{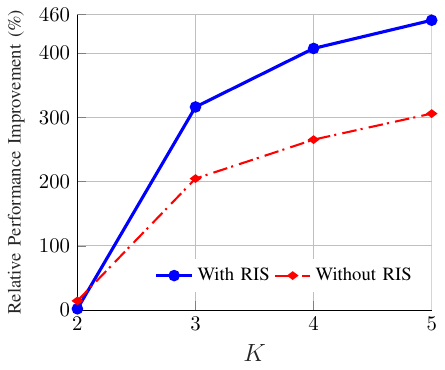}
    \caption{Average max-min rate benefits of utilizing RSMA versus $K$ for {\bf Scenario 1}.    }
	\label{Fig-rr7} 
\end{figure}
Fig. \ref{Fig-rr7} shows the benefits of utilizing RSMA versus $K$ for {\bf Scenario 1}. This figure reveals that the higher the number of users is, the higher gain RSMA can provide. Indeed, the amount of interference escalates, when there are more users in the system, and an interference-management technique becomes more beneficial when the interference level is higher. Therefore, the gains achieved by RSMA increase with $K$.

 Another interesting observation gleaked from Fig. \ref{Fig-rr7} is the impact of RISs on the benefits of RSMA. When the number of BS antennas is lower than the number of users per cell, RISs increase the RSMA gains. However, when $K = N_{BS} = 2$, RISs reduce the RSMA gain and make it almost negligible. In this example, the RSMA gain for $K = 2$ is around 14\% without RISs, which is reduced to around 1.7\% when RISs are used. Note that RISs also have interference-management capabilities in certain scenarios \cite{santamaria2023interference}. Hence, RISs can adequately mitigate the impact of interference when the interference is not strong. By contrast, the RIS has to be supported by more powerful interference-management techniques when the interference is strong. In this case, RISs can improve the coverage, while RSMA is responsible for interference management.

 \subsubsection{Impact of the number of receive antennas} 
 \begin{figure}[t]
    \centering
\includegraphics[width=.4\textwidth]{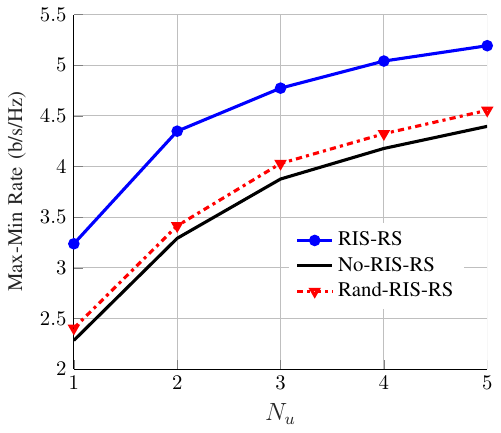}
    \caption{Average max-min rate versus $N_u$ for $N_{BS}=2$, $P = 10$ dB, $L = 1$, $M = 1$, $\epsilon=10^{-5}$,  $n_t=256$ bits, and $N_{{RIS}}=20$.   }
	\label{Fig-rrn1} 
\end{figure} 
 Fig. \ref{Fig-rrn1} shows that the average max-min rate increases with $N_u$, and the highest gain is achieved when $N_u$ is increased from 1 to 2, i.e. when switching from a MISO system to a MIMO one. Indeed, Fig. \ref{Fig-rrn1} illustrates that MIMO systems significantly outperform MISO systems even when there are only two transmit antennas, demonstrating the superiority of our RSMA schemes proposed for MIMO systems over the schemes in \cite{soleymani2023optimization}. Moreover, the gain of MIMO systems scales with the number of receive antennas.

 \begin{figure}[t]
    \centering
\includegraphics[width=.44\textwidth]{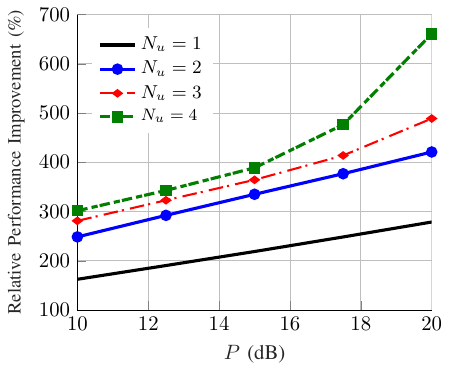}
    \caption{Comparison between multiple-stream and single-stream date transmission for $N_{{BS}}=2$, $K=3$, $L=1$,  $M=1$, $\epsilon=10^{-5}$,  $n=256$ bits and $N_{{RIS}}=20$. 
    }
	\label{Fig-rrn2} 
\end{figure} 
 Fig. \ref{Fig-rrn2} shows the gains achieved by utilizing RSMA in RIS-aided systems versus the power budget at the BS for different $N_u$. Interestingly, the benefits of using RSMA increase with $N_u$, and switching from MISO to MIMO enhances the improvements obtained by using the proposed RSMA scheme in RIS-assisted systems with high user loads. Additionally, the highest relative gain is achieved when $N_u$ increases from 1 to 2, especially at lower SNRs. In Fig. \ref{Fig-rr7}, we showed that RSMA does not provide a considerable gain in MIMO RIS-aided systems when $N_{BS}\geq K$. This means that increasing the number of transmit antennas decreases the RSMA benefits, since the BS can easily manage interference when $N_{BS}\geq K$. However, the number of receive antennas has a different impact on the performance of RSMA when $N_{BS}< K$. Indeed, increasing $N_u$ enhances the effective SINR at the users, but it cannot substantially contribute to interference management when the user load is high, i.e., when $N_{BS}< K$ since the maximum number of streams that the BS can support is still lower than the number of users. Note that increasing the number of antennas either at the receiver side or at the transmitter side enhances system performance, but each could lead to a different impact on the gains achieved by RSMA.

\subsubsection{Comparison with NOMA and TIN} 
\begin{figure}[t]
    \centering
    \begin{subfigure}[t]{0.22\textwidth}
        \centering
           \includegraphics[width=\textwidth]{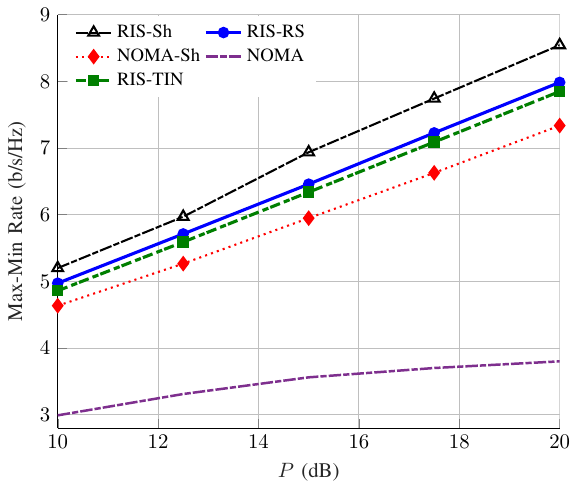}
        \caption{$K=2$.}
    \end{subfigure}
~~
\begin{subfigure}[t]{0.22\textwidth}
        \centering
       \includegraphics[width=\textwidth]{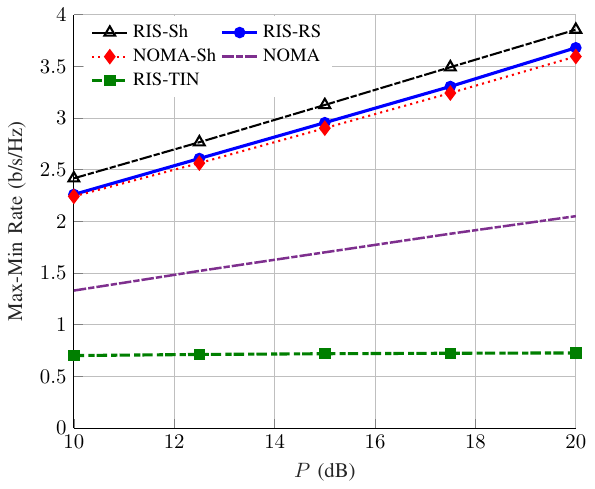}
        \caption{$K=4$.}
    \end{subfigure}%
    \caption{Average max-min rate versus the power budget at the BSs for {\bf Scenario 1}.
}
	\label{Fig-rr8} 
\end{figure}
Fig. \ref{Fig-rr8} shows the average max-min rate versus the power budget at the BSs for {\bf Scenario 1}. Note that the NOMA-based schemes in \cite{soleymani2022noma,soleymani2023noma} consider only the Shannon rates. Hence, we provide the average results for the Shannon and FBL rates achieved by the schemes in \cite{soleymani2022noma, soleymani2023noma}.  {Interestingly, the TIN scheme with FBL coding outperforms the NOMA-based scheme with both the Shannon and FBL rates, when $K = 2$. However, when $K = 4$, the NOMA-based scheme with either FBL rates or Shannon rates noticeably outperforms the TIN scheme. As discussed in the results of Fig. \ref{Fig-rr7}, when $K = 2$, the interference level is low, and TIN performs very close to the RSMA scheme. However, when $K$ increases, the interference becomes more severe, which makes TIN highly suboptimal. Additionally, Fig. \ref{Fig-rr8} illustrates that the RSMA scheme outperforms the other schemes in both cases, i.e., $K = 2$ and $K = 4$. Since RSMA includes both TIN and SIC strategies, it can switch between these strategies, depending on the interference level. Thus, RSMA performs better than the TIN and the NOMA-based schemes.} We can also observe in Fig. \ref{Fig-rr8} that applying the solutions associated with the Shannon rate to the FBL rates drastically decreases the performance. Indeed, the performance gap of the schemes proposed in \cite{soleymani2022noma, soleymani2023noma} for the FBL and Shannon rates is significantly higher than the gap for the RSMA scheme that considers FBL rates. As the schemes in \cite{soleymani2022noma, soleymani2023noma} are designed for Shannon rates, their performance significantly drops by switching from the Shannon rate to the FBL rate, which shows the importance of developing resource allocation schemes considering the FBL rate.

\subsubsection{Comparison with single-stream RSMA} 

\begin{figure}[t]
    \centering
       \includegraphics[width=.44\textwidth]{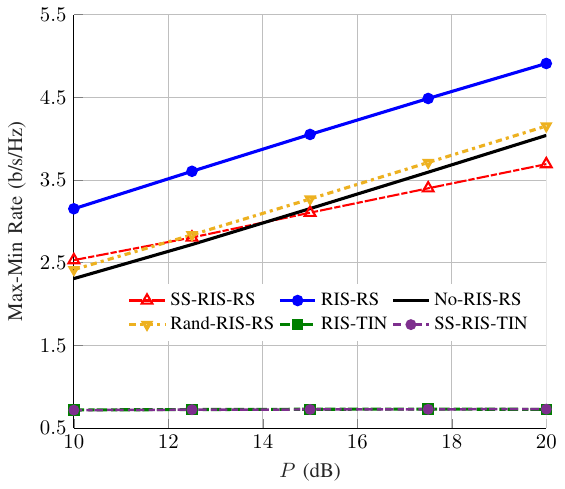}
    \caption{ {Comparison between multiple-stream and single-stream date transmission for $N_{{BS}}=2$, $N_{u}=2$, $K=4$, $L=1$,  $M=1$, $\epsilon=10^{-5}$,  $n=256$ bits and $N_{{RIS}}=20$.}
    }
	\label{Fig-rrn3} 
\end{figure} 
 In Fig. \ref{Fig-rrn3}, we compare the performance of our RSMA schemes proposed for single- and multi-stream data transmission per user. In this figure, the maximum number of streams per user is $2$. However, our RSMA scheme supporting multiple-stream data transmission significantly outperforms the one using single-stream data transmission, especially at higher SNRs. Indeed, the RSMA scheme conceived for systems operating without RIS outperforms single-stream RSMA in RIS-aided systems. In other words, the additional gains achieved by using RIS would be eroded if we employed single-stream data transmission. Note that the single-stream scheme is an enhanced version of the RSMA schemes in \cite{soleymani2023optimization} developed for MISO systems.  

\subsection{Energy Efficiency Metrics}\label{sec-iv-b}
Here, we evaluate the performance of the proposed RSMA scheme in terms of the max-min EE. Moreover, we investigate how RSMA and RISs influence each other from an EE perspective. Specifically, we evaluate the impact of $p_c$, $n$, and $\epsilon$ on the performance of RSMA and RISs. In this subsection, we adopt $P = 10$ dB for all figures. Moreover, for a fair comparison between the RIS-aided systems and the systems without an RIS, we consider a lower static power for the systems operating without RISs. Specifically, we assume that 1 Watt is required to operate each RIS. Thus, $p_{c,No-RIS} = p_c -1/K$, where $p_c$ is the static power consumed by the system, which is defined as in \eqref{(15)}. Even with this assumption, our results show that RISs can improve
the EE, as discussed below. 
\begin{figure}[t]
    \centering
           \includegraphics[width=.44\textwidth]{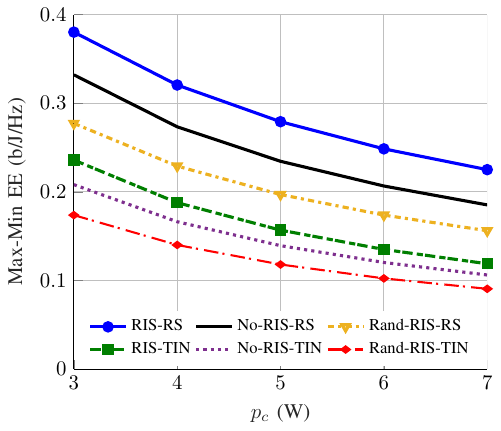}
    \caption{Average max-min EE versus $p_c$ for {\bf Scenario 1}. 
    }
	\label{Fig-ee1} 
\end{figure}
\subsubsection{Impact of $p_c$} 
We illustrate the average max-min EE versus $p_c$ for {\bf Scenario 1} in Fig. \ref{Fig-ee1}. This example reveals that each RSMA scheme significantly outperforms all the TIN schemes both with and without RIS. Moreover, an RIS can improve the performance substantially only when its coefficients are accurately optimized. Additionally, the RSMA scheme of RIS-aided systems performs the best among the schemes considered.

\begin{figure}[t]
    \centering
    \begin{subfigure}[t]{0.24\textwidth}
        \centering
           \includegraphics[width=\textwidth]{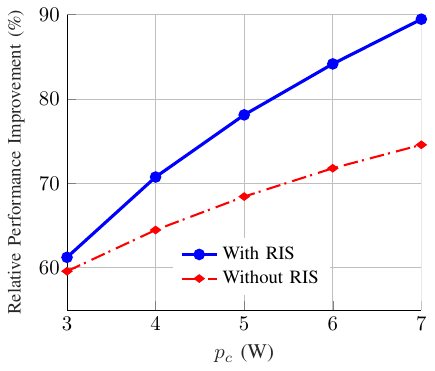}
        \caption{Benefits of RSMA over SDMA.}
    \end{subfigure}
\begin{subfigure}[t]{0.24\textwidth}
        \centering
       \includegraphics[width=\textwidth]{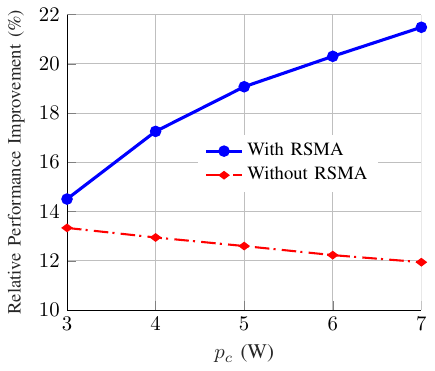}
        \caption{Benefits of RIS over systems operating without RISs.}
    \end{subfigure}%
    \caption{Average EE benefits of utilizing RSMA and RISs versus $p_c$ for {\bf Scenario 1}. 
    }
	\label{Fig-ee2} 
\end{figure}
Fig. \ref{Fig-ee2} shows the EE gains of RSMA and RISs versus $p_c$ for {\bf Scenario 1}. Again, in this example, we can observe that RSMA and RISs mutually enhance each other's gains. Hence, they are mutually beneficial technologies. Additionally,
the gain of RSMA is much higher than the gain of RISs in this example.

\subsubsection{Impact of the reliability constraint}
\begin{figure}[t]
        \centering
           \includegraphics[width=.44\textwidth]{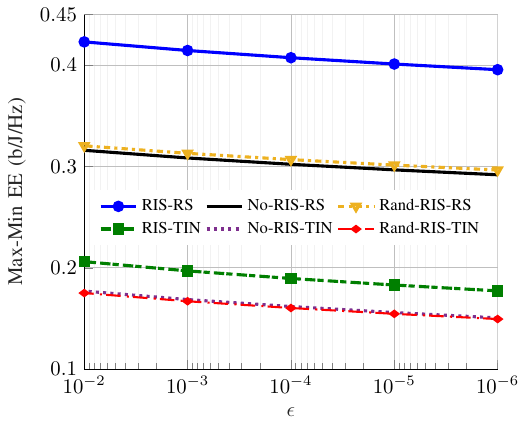}
    \caption{Average max-min EE  versus $\epsilon$ for {\bf Scenario 2}. 
    }
	\label{Fig-ee3} 
\end{figure}
Fig. \ref{Fig-ee3}  illustrates
the average max-min EE versus $\epsilon$ for {\bf Scenario 2}. This figure shows that RSMA with RISs can significantly improve the EE. Moreover, similar to the other examples of this section, all the RSMA designs outperform each TIN scheme both with and without employing an RIS. Additionally, the RSMA scheme of RIS-aided systems substantially outperforms all the other
schemes.

\begin{figure}[t]
    \centering
    \begin{subfigure}[t]{0.24\textwidth}
        \centering
           \includegraphics[width=\textwidth]{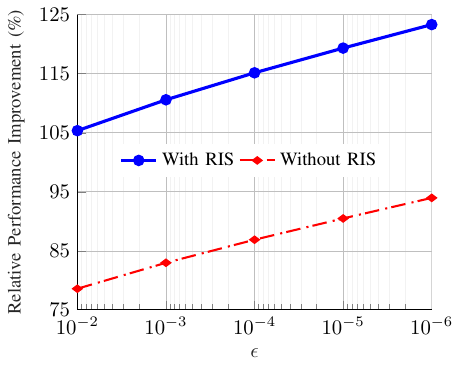}
        \caption{Benefits of RSMA over SDMA.}
    \end{subfigure}
\begin{subfigure}[t]{0.24\textwidth}
        \centering
       \includegraphics[width=\textwidth]{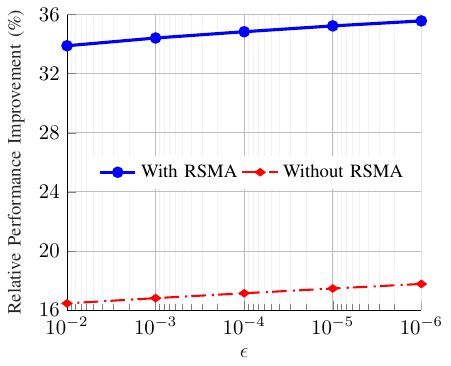}
        \caption{Benefits of RIS over systems operating without RISs.}
    \end{subfigure}%
    \caption{Average EE benefits of utilizing RSMA and RISs versus $\epsilon$ for {\bf Scenario 2}. 
    }
	\label{Fig-ee4} 
\end{figure}
Fig. \ref{Fig-ee4} shows the gains of employing RSMA and RISs in
terms of the max-min EE versus $\epsilon$ for {\bf Scenario 2}. Here, the gain along with RSMA is more significant in RIS-aided systems than in systems operating without an RIS. Moreover, the gains of RSMA and/or RISs are enhanced when $\epsilon$ is reduced. Thus, RSMA and RISs are more beneficial when highly reliable communication is needed. Additionally, we can observe that RSMA enhances the gains of RISs and vice versa, which implies again that they are mutually beneficial technologies.

\subsubsection{Impact of  the packet length}
\begin{figure}[t]
    \centering
        \includegraphics[width=.44\textwidth]{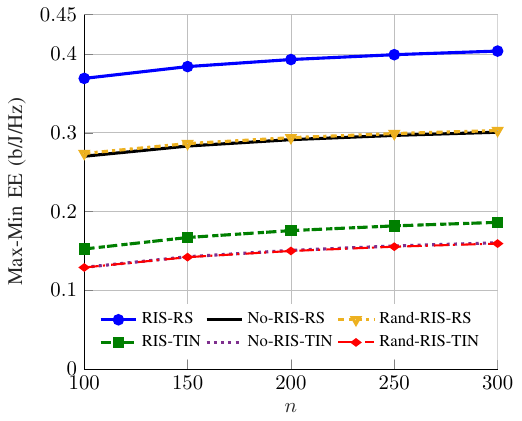}
    \caption{Average max-min EE versus versus $n$ for {\bf Scenario 2}. 
    }
	\label{Fig-ee5} 
\end{figure}
Fig. \ref{Fig-ee5} shows the average
max-min EE versus the packet length in bits for {\bf Scenario 2}. In this example, RSMA provides higher benefits compared to an RIS. Specifically, a system without an RIS using RSMA outperforms an RIS-assisted system with TIN. Additionally, an RIS only improves the performance if its coefficients are appropriately optimized.

\begin{figure}[t]
    \centering
    \begin{subfigure}[t]{0.24\textwidth}
        \centering
           \includegraphics[width=\textwidth]{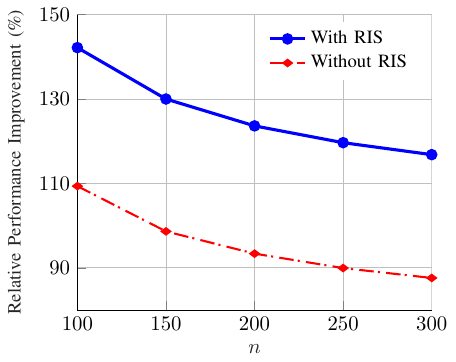}
        \caption{Benefits of RSMA over SDMA.}
    \end{subfigure}
\begin{subfigure}[t]{0.24\textwidth}
        \centering
       \includegraphics[width=\textwidth]{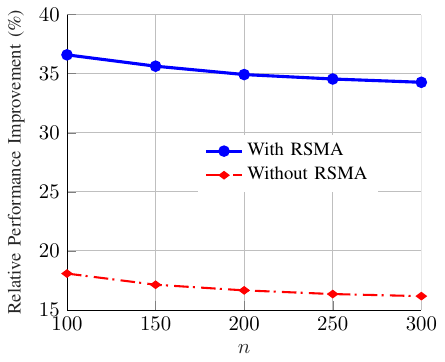}
        \caption{Benefits of RIS over systems operating without RISs.}
    \end{subfigure}%
    \caption{Average EE benefits of utilizing RSMA and RISs versus $n$ for {\bf Scenario 2}. 
    }
	\label{Fig-ee6} 
\end{figure}
Fig. \ref{Fig-ee6} shows the benefits of employing RSMA and RISs in terms of the max-min EE versus $n$ for {\bf Scenario 2}. As it can be observed, RSMA provides higher gains in an RIS-aided system than in a system operating without RISs, and the RSMA gains are higher when the packet length is reduced. Similarly, the RIS gains increase when we use RSMA in this example, since the system is interference-limited, and an RIS alone cannot completely eliminate the interference.

\section{Summary and Conclusions}\label{sec-v}
This paper developed RSMA algorithms for maximizing the minimum weighted rate and/or EE of the users in a multicell RIS-aided URLLC MIMO BC. The main findings of this study are summarized below:
\begin{itemize}
\item RSMA and RISs substantially enhance both the SE and EE of MU-MIMO URLLC systems. The SE gains of RISs and RSMA are higher than their EE benefits, since the optimal EE usually favors a lower transmit power than the optimal SE. Moreover, the benefits of RSMA and RISs are increasing functions of the BSs transmission
power.

\item RISs impact the performance of RSMA differently, depending on the user load. When the number of users is higher than the number of TAs at the BSs, RISs improve the RSMA gains and vice versa. However, when the user load is low, RISs erase the RSMA gains.

\item  RISs and RSMA can provide higher gains when the reliability and latency constraints are more demanding. Particularly, the benefits of RSMA and RIS are enhanced when the maximum tolerable packet error rate and packet length are reduced.

\item RSMA provides higher SE and EE gains than RISs when the user load is higher than one. Additionally, in the scenarios considered, RSMA always outperformed SDMA both with and without RISs.

\end{itemize}
\appendices
\section{Proof of Lemma \ref{lem-4}}\label{app-a}
The rates $r_{c,lk}$ and $r_{p,lk}$ comprise a term associated with the first-order Shannon rate and a term with the channel dispersion
expression. Since both rates $r_{c,lk}$ and $r_{p,lk}$ have a similar structure, we prove that the lower bound holds for $r_{c,lk}$. It is straightforward to show that \eqref{(22)} holds based on the proof of the lower bound in \eqref{(23)}. To calculate the concave lower bounds, we derive a lower bound for each part separately. Specifically, we utilize Lemma \ref{lem-2} to compute a lower
bound for the Shannon rate part as
\begin{multline}\label{(37)}
    \ln\left|{\bf I} + {\bf D}^{-1}_{c,lk} {\bf S}_{c,lk}  \right|
    \geq \ln\left|{\bf I} + \bar{\bf D}^{-1}_{c,lk} \bar{\bf S}_{c,lk}  \right|
    - \text{Tr}\left(  \bar{\bf D}^{-1}_{c,lk} \bar{\bf S}_{c,lk} \right)
    \\
    + 2\mathfrak{R}\left\{\bar{\bf W}_{l}^H\bar{\bf H}_{lk,l}^H \bar{\bf D}^{-1}_{c,lk} {\bf H}_{lk,l} {\bf W}_{l} \right\}
    \\
    -\!\text{Tr}\!\!\left(\!\! ( \bar{\bf D}^{-1}_{c,lk} \!\!-\! (\bar{\bf S}_{c,lk}\! +\! \bar{\bf D}^{-1}_{c,lk})^{-1}) 
    (\bar{\bf H}_{lk,l} {\bf W}_{l} {\bf W}_{l}^H \bar{\bf H}_{lk,l}^H \!+\! {\bf D}_{c,lk})\!\!
    \right)\!.
\end{multline}
The part of $r_{c,lk}$, which is related to the channel dispersion term, is 
$-Q^{-1}(\epsilon_c) \sqrt{\frac{\zeta_{c,lk}}{n}}$, 
where 
$\zeta_{c,lk} = 2\text{Tr} \left({\bf S}_{c,lk}({\bf S}_{c,lk}+{\bf D}_{c,lk})^{-1}\right)$ 
is the channel dispersion for decoding the common message at U$_{lk}$, and $\epsilon_c$ and n are constant.
Thus, deriving a concave lower bound for $-Q^{-1}(\epsilon_c)
\sqrt{\frac{\zeta_{c,lk}}{n}}$ is equivalent to finding a convex upper bound for $\sqrt{\zeta_{c,lk}}$. To obtain a convex upper bound for $\sqrt{\zeta_{c,lk}}$, we first employ the inequality below:
\begin{equation}\label{(38)}
    \sqrt{\zeta_{c,lk}}\leq \frac{\sqrt{\bar{\zeta}_{c,lk}}}{2}+\frac{\zeta_{c,lk}}{2\sqrt{\bar{\zeta}_{c,lk}}},
\end{equation}
where $\bar{\zeta}_{c,lk}$ is defined as in Lemma \ref{lem-4}. Unfortunately, $\zeta_{c,lk}$ is not a convex function, but we can derive a convex upper bound for $\zeta_{c,lk}$ upon utilizing the bound in Lemma \ref{lem-3}. To this end, we rewrite $\zeta_{c,lk}$ as
\begin{equation}\label{(39)}
    \zeta_{c,lk}=2\text{Tr}\left({\bf I}- {\bf D}_{c,lk}({\bf S}_{c,lk}+{\bf D}_{c,lk})^{-1} \right).
\end{equation}
Now, we want to obtain a concave lower bound for
$\text{Tr}\left({\bf D}_{c,lk}({\bf S}_{c,lk}+{\bf D}_{c,lk})^{-1} \right)$, which can be derived by utilizing Lemma \ref{lem-3} as
\begin{multline}\label{(40)}
   \text{Tr}\left({\bf D}_{c,lk}({\bf S}_{c,lk}+{\bf D}_{c,lk})^{-1} \right)\geq  
   \\
   2\sum_{ij}\mathfrak{R}\left\{\text{Tr}\left((\bar{\bf S}_{c,lk} + \bar{\bf D}_{c,lk})^{-1}  
   \bar{\bf H}_{lk,i} \bar{\bf W}_{ij} {\bf W}_{ij}^H \bar{\bf H}_{lk,i}^H
   \right)\right\}
   \\
   + 2\sum_{i\neq l}\mathfrak{R}\left\{\text{Tr}\left((\bar{\bf S}_{c,lk} + \bar{\bf D}_{c,lk})^{-1}   
   \bar{\bf H}_{lk,i} \bar{\bf W}_{i} {\bf W}_{i}^H \bar{\bf H}_{lk,i}^H
   \right)\right\}
   \\
   -\!\text{Tr}\!\left[\! (\bar{\bf S}_{c,lk}\! + \!\bar{\bf D}_{c,lk})^{-1} 
   \bar{\bf D}_{c,lk}
   (\bar{\bf S}_{c,lk}\! +\! \bar{\bf D}_{c,lk})^{-1} 
    ({\bf S}_{c,lk}\! +\! {\bf D}_{c,lk})
   \!\right]\!\!.
\end{multline}
From \eqref{(38)}, \eqref{(39)}, and \eqref{(40)}, we have
\begin{multline}\label{(41)}
    -Q^{-1}(\epsilon_c)
\sqrt{\frac{\zeta_{c,lk}}{n}}\!\geq \!
-Q^{-1}(\epsilon_c)
\sqrt{\frac{\bar{\zeta}_{c,lk}}{4n}}
-Q^{-1}(\epsilon_c)
\frac{\text{Tr}({\bf I})}{\sqrt{n\bar{\zeta}_{c,lk}}}
\\
+\!\frac{2Q^{-1}(\epsilon_c)}{\sqrt{n\bar{\zeta}_{c,lk}}}
\!\!\sum_{ij}\mathfrak{R}\!\left\{\!\text{Tr}\left((\bar{\bf S}_{c,lk} \!+\! \bar{\bf D}_{c,lk})^{-1}  
   \bar{\bf H}_{lk,i} \bar{\bf W}_{ij} {\bf W}_{ij}^H \bar{\bf H}_{lk,i}^H
   \!\right)\!\right\}
\\
+\!\frac{2Q^{-1}(\epsilon_c)}{\sqrt{n\bar{\zeta}_{c,lk}}}
\!\!\sum_{i\neq l}\!\mathfrak{R}\!\left\{\!\text{Tr}\left((\bar{\bf S}_{c,lk} \!+\! \bar{\bf D}_{c,lk})^{-1}   
   \bar{\bf H}_{lk,i} \bar{\bf W}_{i} {\bf W}_{i}^H \bar{\bf H}_{lk,i}^H
   \!\right)\!\right\}
\\
-\frac{Q^{-1}(\epsilon_c)}{\sqrt{n\bar{\zeta}_{c,lk}}}
\text{Tr}\!\left[\! (\bar{\bf S}_{c,lk}\! + \!\bar{\bf D}_{c,lk})^{-1} 
   \bar{\bf D}_{c,lk}
   (\bar{\bf S}_{c,lk}\! +\! \bar{\bf D}_{c,lk})^{-1} 
   \right.   \\   \left. \times
    ({\bf S}_{c,lk}\! +\! {\bf D}_{c,lk})
   \right].
\end{multline}
Upon using \eqref{(37)} and \eqref{(41)}, we can readily obtain \eqref{(23)}.
\bibliographystyle{IEEEtran}
\bibliography{ref2}

\begin{IEEEbiography}[{\includegraphics[width=1in,height=1.25in,clip,keepaspectratio]{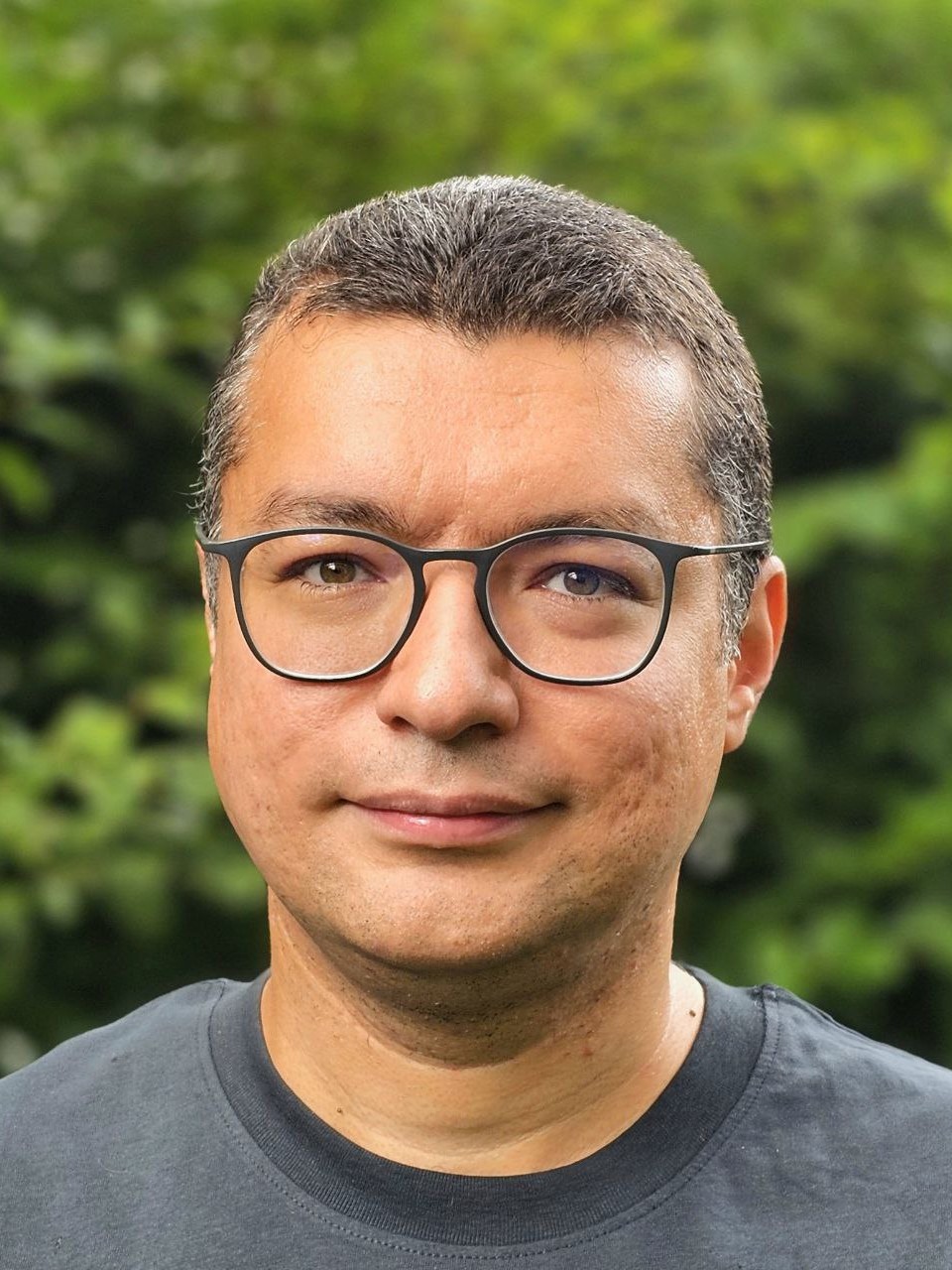}}]
{Mohammad Soleymani} was born in Arak, Iran. He received the B.Sc. degree from Amirkabir University of Technology (Tehran Polytechnic), the M.Sc. degree from Sharif University of Technology, Tehran, Iran, and the Ph.D. degree (with distinction) in electrical engineering from the University of Paderborn, Germany. He is currently an Akademischer Rat a. Z. with the Signal and System Theory Group at the University of Paderborn. He was a Visiting Researcher at the University of Cantabria, Spain. He serves on the editorial boards of ELSEVIER Signal Processing, EURASIP Journal on Wireless Communications and Networking, and Springer Journal of Wireless Personal Communications. His research interests include multi-user MIMO, wireless networking, numerical optimization, ultra-reliable low-latency communications (URLLC), and statistical signal processing.
\end{IEEEbiography}

\begin{IEEEbiography}[{\includegraphics[width=1in,height=1.25in,clip,keepaspectratio]{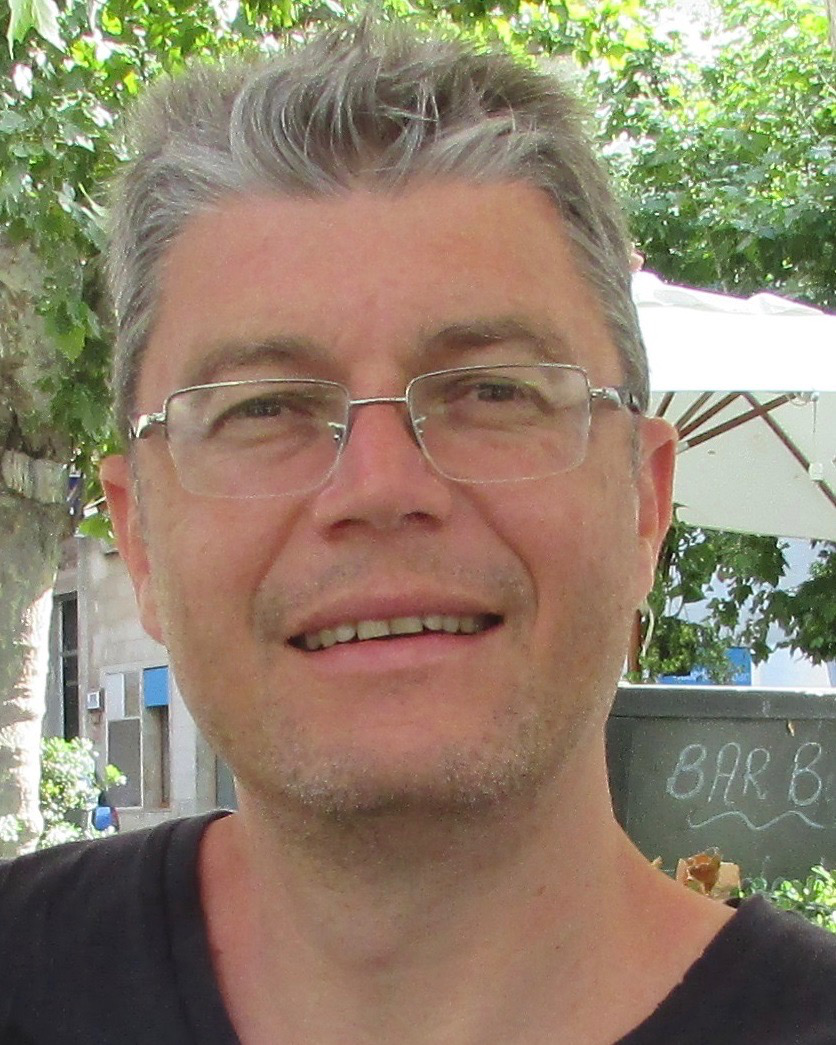}}]{Ignacio Santamaria} (M'96-SM'05) 
received the Telecommunication Engineer degree and the Ph.D. degree in electrical engineering from the Universidad Politecnica de Madrid (UPM), Spain, in 1991 and 1995, respectively. In 1992, he joined the Department of Communications Engineering, Universidad de Cantabria, Spain, where he is Full Professor since 2007. He has co-authored more than 250 publications in refereed journals and international conference papers, and holds two patents. He has co-authored the book D. Ramirez, I. Santamaria, and L.L. Scharf, ``Coherence in Signal Processing and Machine Learning'', Springer, 2022. His current research interests include signal processing algorithms and information-theoretic aspects of multiuser multiantenna wireless communication systems, multivariate statistical techniques and machine learning theories. He has been involved in numerous national and international research projects on these topics.  He has been a visiting researcher at the University of Florida (in 2000 and 2004), at the University of Texas at Austin (in 2009), and at the Colorado State University (in 2015 and 2017). He has been Associate Editor of the IEEE Transactions on Signal Processing (2011-2015), and Senior Area Editor of the IEEE Transactions on Signal Processing (2013-2015). He has been a member of the IEEE Machine Learning for Signal Processing Technical Committee (2009-2014), member of the IEEE Signal Processing Theory and Methods Technical Committee (2020-2022), and member of the IEEE Data Science Initiative (DSI) steering committee (2020-2022). Prof. Santamaria co-authored a paper that received the 2012 IEEE Signal Processing Society Young Author Best Paper Award, and has received the 2022 IHP International Wolfgang Mehr Fellowship Award.
\end{IEEEbiography} 

\begin{IEEEbiography}[{\includegraphics[width=1in,height=1.25in,clip,keepaspectratio]{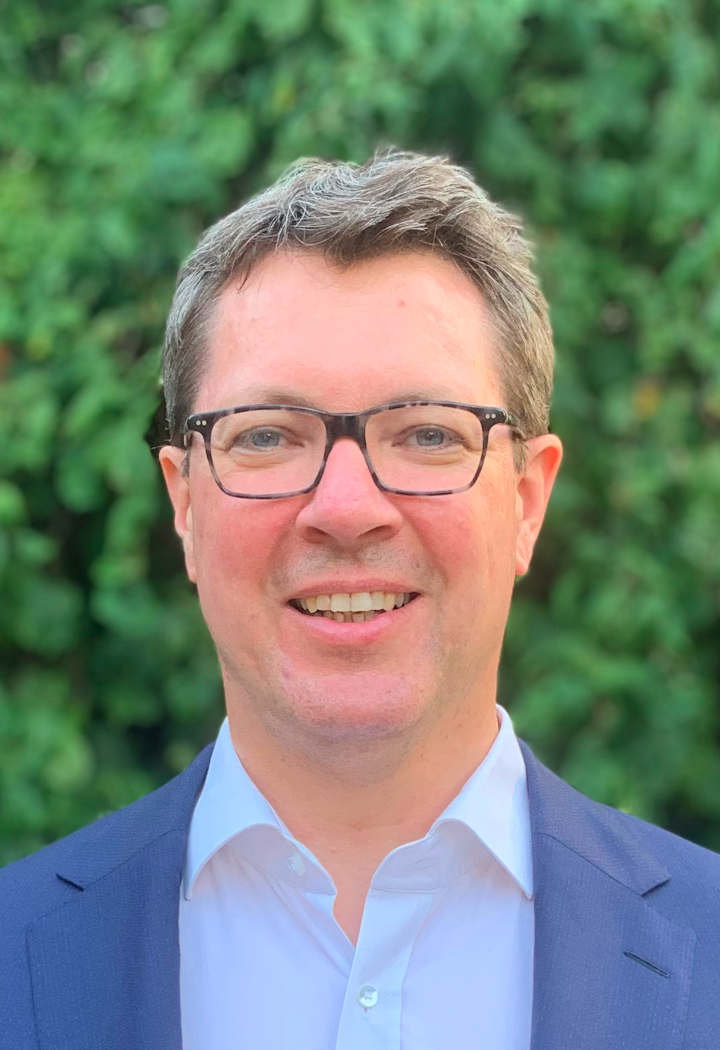}}]
{Eduard Axel Jorswieck} (Fellow, IEEE)  is full professor with the Faculty of Electrical Engineering, Information Technology, Physics of TU Braunschweig. He is the managing director of the Institute for Communications Technology and Full Professor at TU Braunschweig, Germany. From 2008 until 2019, he was the head of the Chair for Communications Theory and Full Professor at Dresden University of Technology, Germany. He is IEEE Fellow. His general interests are in signal processing for communications and networking, applied information theory and communication theory. His research interests include multiple antenna communications, wireless interference networks, reliability and resilience, and physical layer security. He has published 200 journal articles, 19 book chapters, one book, four monographs, and some 350 conference papers. He was a co-recipient of the IEEE Signal Processing Society Best Paper Award. He and his colleagues were also recipients of the Best Paper Awards and the Best Student Paper Awards from the IEEE CAMSAP 2011, IEEE WCSP 2012, IEEE SPAWC 2012, IEEE ICUFN 2018, PETS 2019, and ISWCS 2019, and IEEE ICC 2024. Since 2017, he has been the Editor-in-Chief of the EURASIP Journal on Wireless Communications and Networking. Since 2024, he has been an Editor for IEEE TRANSACTIONS ON INFORMATION THEORY. He was on the editorial boards of the IEEE SIGNAL PROCESSING LETTERS, IEEE TRANSACTIONS ON SIGNAL PROCESSING, the IEEE TRANSACTIONS ON WIRELESS COMMUNICATIONS, IEEE TRANSACTIONS ON INFORMATION FORENSICS AND SECURITY, and IEEE TRANSACTIONS ON COMMUNICATIONS. He received the 2019 outstanding editorial board award from the IEEE Transactions on Information Forensics and Security. 
\end{IEEEbiography}

\begin{IEEEbiography}[{\includegraphics[width=1in,height=1.25in,clip,keepaspectratio]{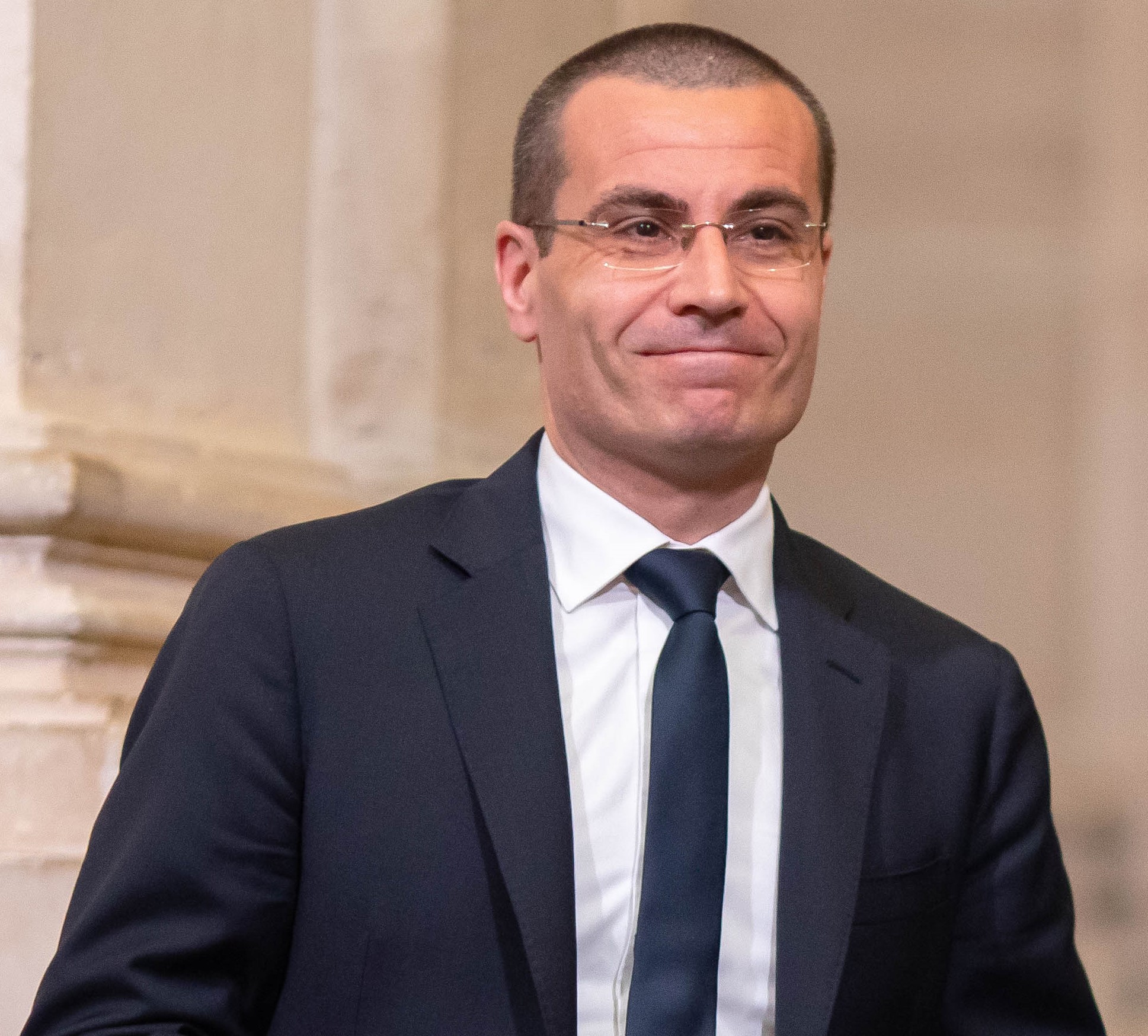}}]
{Marco Di Renzo} (Fellow, IEEE) received the Laurea (cum laude) and Ph.D. degrees in electrical engineering from the University of L’Aquila, Italy, in 2003 and 2007, respectively, and the Habilitation à Diriger des Recherches (Doctor of Science) degree from University Paris-Sud (currently Paris-Saclay University), France, in 2013. Currently, he is a CNRS Research Director (Professor) and the Head of the Intelligent Physical Communications group with the Laboratory of Signals and Systems (L2S) at CNRS \& CentraleSupélec, Paris-Saclay University, Paris, France, as well as a Chair Professor in Telecommunications Engineering with the Centre for Telecommunications Research -- Department of Engineering, King’s College London, London, United Kingdom. He was a France-Nokia Chair of Excellence in ICT at the University of Oulu (Finland), a Tan Chin Tuan Exchange Fellow in Engineering at Nanyang Technological University (Singapore), a Fulbright Fellow at The City University of New York (USA), a Nokia Foundation Visiting Professor at Aalto University (Finland), and a Royal Academy of Engineering Distinguished Visiting Fellow at Queen’s University Belfast (U.K.). He is a Fellow of the IEEE, IET, EURASIP, and AAIA; an Academician of AIIA; an Ordinary Member of the European Academy of Sciences and Arts, an Ordinary Member of the Academia Europaea, an Ordinary Member of the Italian Academy of Technology and Engineering; an Ambassador of the European Association on Antennas and Propagation; and a Highly Cited Researcher. His recent research awards include the Michel Monpetit Prize conferred by the French Academy of Sciences, the IEEE Communications Society Heinrich Hertz Award, and the IEEE Communications Society Marconi Prize Paper Award in Wireless Communications. He served as the Editor-in-Chief of IEEE Communications Letters from 2019 to 2023. His current main roles within the IEEE Communications Society include serving as a Voting Member of the Fellow Evaluation Standing Committee, as the Chair of the Publications Misconduct Ad Hoc Committee, and as the Director of Journals. Also, he is on the Editorial Board of the Proceedings of the IEEE.
\end{IEEEbiography}

\begin{IEEEbiography}[{\includegraphics[width=1in,height=1.25in,clip,keepaspectratio]{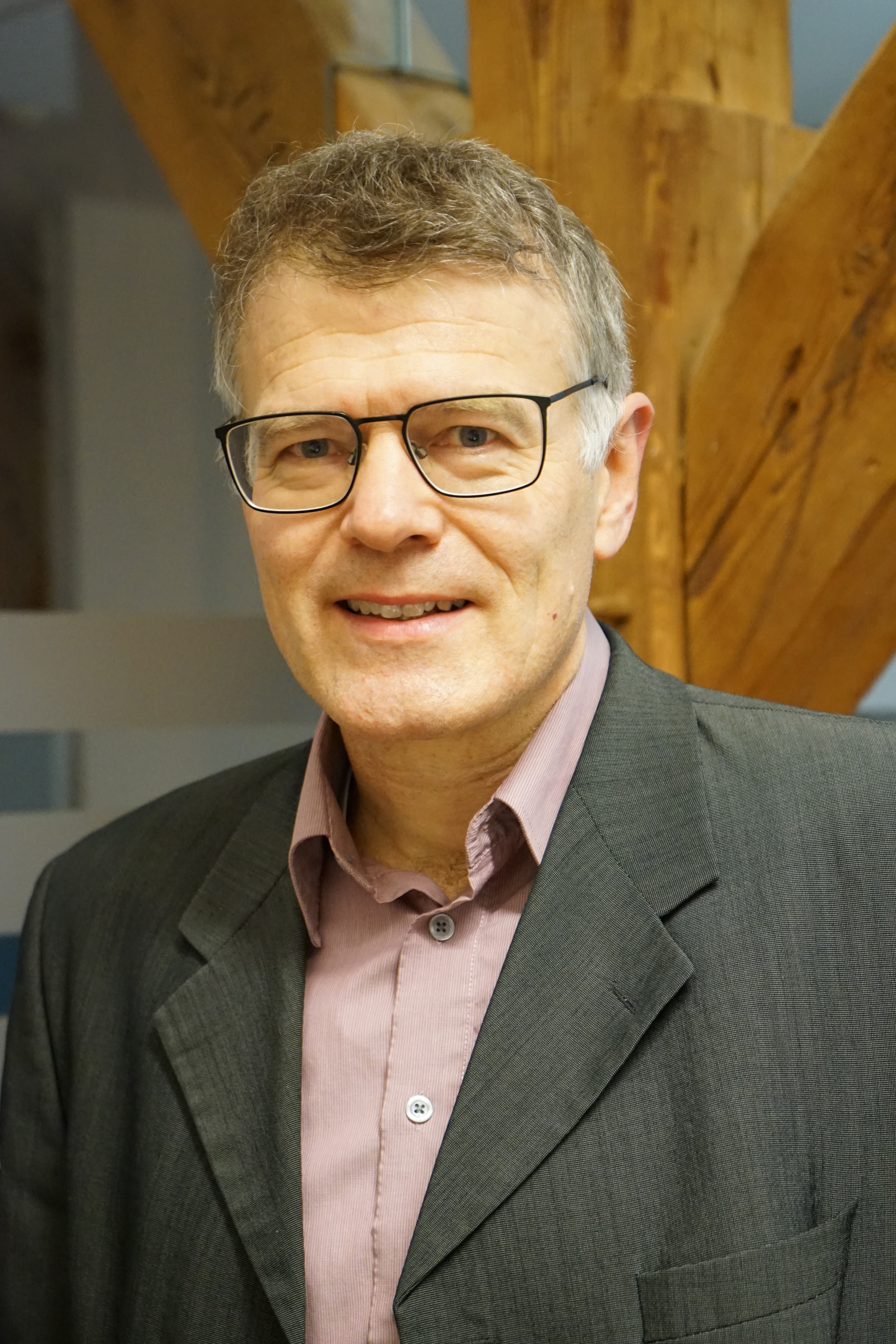}}]
{Robert Schober} (S'98, M'01, SM'08, F'10) received the Diplom (Univ.) and the Ph.D. degrees in electrical engineering from Friedrich-Alexander University of Erlangen-Nuremberg (FAU), Germany, in 1997 and 2000, respectively. From 2002 to 2011, he was a Professor and Canada Research Chair at the University of British Columbia (UBC), Vancouver, Canada. Since January 2012 he is an Alexander von Humboldt Professor and the Chair for Digital Communication at FAU. His research interests fall into the broad areas of Communication Theory, Wireless and Molecular Communications, and Statistical Signal Processing.

Robert received several awards for his work including the 2002 Heinz Maier Leibnitz Award of the German Science Foundation (DFG), the 2004 Innovations Award of the Vodafone Foundation for Research in Mobile Communications, a 2006 UBC Killam Research Prize, a 2007 Wilhelm Friedrich Bessel Research Award of the Alexander von Humboldt Foundation, the 2008 Charles McDowell Award for Excellence in Research from UBC, a 2011 Alexander von Humboldt Professorship, a 2012 NSERC E.W.R. Stacie Fellowship, a 2017 Wireless Communications Recognition Award by the IEEE Wireless Communications Technical Committee, and the 2022 IEEE Vehicular Technology Society Stuart F. Meyer Memorial Award. Furthermore, he received numerous Best Paper Awards for his work including the 2022 ComSoc Stephen O. Rice Prize and the 2023 ComSoc Leonard G. Abraham Prize. Since 2017, he has been listed as a Highly Cited Researcher by the Web of Science. Robert is a Fellow of the Canadian Academy of Engineering, a Fellow of the Engineering Institute of Canada, and a Member of the German National Academy of Science and Engineering.

He served as Editor-in-Chief of the IEEE Transactions on Communications, VP Publications of the IEEE Communication Society (ComSoc), ComSoc Member at Large, and ComSoc Treasurer. Currently, he serves as Senior Editor of the Proceedings of the IEEE and as ComSoc President.

\end{IEEEbiography} 

\begin{IEEEbiography}[{\includegraphics[width=1in,height=1.25in,clip,keepaspectratio]{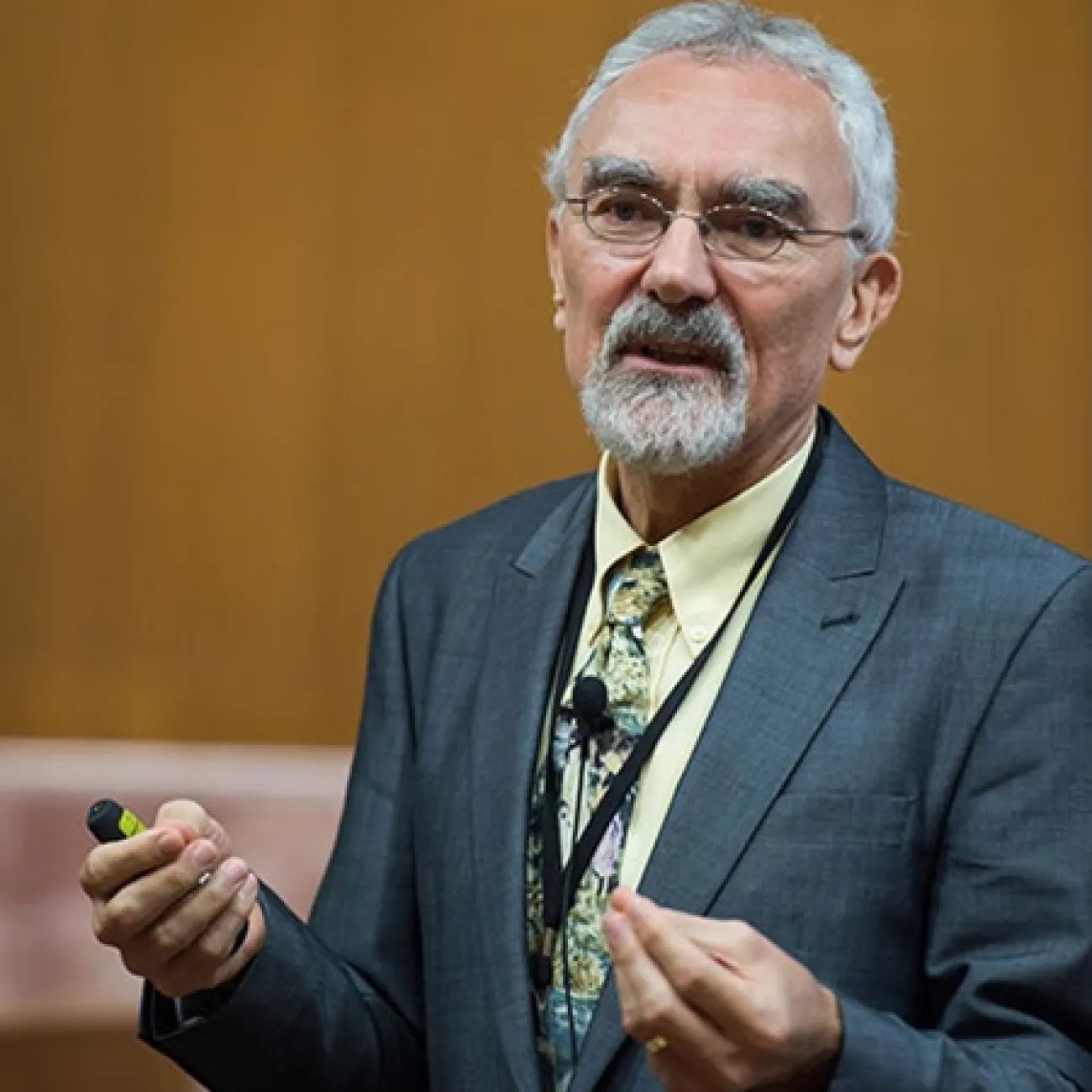}}]
{Lajos Hanzo}  is a Fellow of the Royal Academy of Engineering, FIEEE, FIET, Fellow of EURASIP and a Foreign Member of the Hungarian Academy of Sciences. He coauthored 2000+ contributions at IEEE Xplore and 19 Wiley-IEEE Press monographs. He was bestowed upon the IEEE Eric Sumner Technical Field Award.

\end{IEEEbiography} 

\end{document}